\documentclass[11pt,a4paper]{amsart}
\usepackage{amssymb}
\usepackage{amsfonts}
\usepackage{amsmath}
\usepackage{amsthm}
\usepackage{graphicx}
\usepackage[dvips]{epsfig}
\usepackage{longtable}
\usepackage{graphics}
\usepackage{color}

\usepackage{graphicx}				
\usepackage{amssymb}

\newtheorem{lemma}{Lemma}
\newtheorem{proposition}{Proposition}
\newtheorem{corollary}{Corollary}
\newtheorem{definition}{Definition}
\newtheorem{assumption}{Assumption}
\newtheorem{example}{Example}
\newtheorem{remark}{Remark}

\newcommand{\trans}{{\mathrm{T}}}
\newcommand{\pv}{{\mathrm{pv}}}
\newcommand{\tv}{{\mathrm{tv}}}

\newcommand{\E}[1]{\mathbb{E}\left[ #1 \right]}
\newcommand{\U}[1]{\mathbb{U}\left[ #1 \right]}
\newcommand{\V}[1]{\mathbb{V}\left[ #1 \right]}
\newcommand{\VaR}{\operatorname{VaR}}
\newcommand{\ess}{\operatorname{ess}}
\newcommand{\sign}{\operatorname{sign}}
\newcommand{\ES}{\operatorname{ES}}

\renewcommand{\P}{\mathbb{P}}
\newcommand{\R}{\mathbb{R}}

\newcommand{\Var}{\operatorname{Var}}
\newcommand{\Cov}{\operatorname{Cov}}
\newcommand{\SCR}{\operatorname{SCR}}
\newcommand{\RM}{\operatorname{RM}}
\newcommand{\BE}{\operatorname{BE}}
\newcommand{\TP}{\operatorname{TP}}
\newcommand{\CoC}{\operatorname{CoC}}
\renewcommand{\phi}{\varphi}

\newcommand{\calF}{\mathcal{F}}
\newcommand{\calG}{\mathcal{G}}
\newcommand{\calH}{\mathcal{H}}
\newcommand{\filF}{\mathbb{F}}
\newcommand{\filG}{\mathbb{G}}
\newcommand{\calA}{\mathcal{A}}
\newcommand{\calL}{\mathcal{L}}
\newcommand{\calR}{\mathcal{R}}

\title[Multi-period cost-of-capital valuation]{Insurance valuation: a computable multi-period cost-of-capital approach}
\author{Hampus Engsner, Mathias Lindholm, Filip Lindskog}
\address{Department of Mathematics\\ Stockholm University\\ SE-106 91 Stockholm\\ Sweden}
\date{July 14, 2016}							

\begin{document}
\maketitle

\begin{abstract}
We present an approach to market-consistent multi-period valuation of insurance liability cash flows based on a two-stage valuation procedure. 
First, a portfolio of traded financial instrument aimed at replicating the liability cash flow is fixed. 
Then the residual cash flow is managed by repeated one-period replication using only cash funds.
The latter part takes capital requirements and costs into account, as well as limited liability and risk averseness of capital providers. The cost-of-capital margin is the value of the residual cash flow. We set up a general framework for the cost-of-capital margin and relate it to dynamic risk measurement. 
Moreover, we present explicit formulas and properties of the cost-of-capital margin under further assumptions on the model for the liability cash flow and on the conditional risk measures and utility functions. Finally, we highlight computational aspects of the cost-of-capital margin, and related quantities, in terms of an example from life insurance.

\smallskip
\noindent
{\bf Keywords:} valuation of insurance liabilities, multi-period valuation, market-consistent valuation, cost of capital, risk margin, dynamic risk measurement
\end{abstract}


\section{Introduction}

The current solvency regulatory framework Solvency II emphasizes market-consistent valuation of liabilities; it is explicitly stated that liabilities should be 
``valued at the amount for which they could be ... transferred or settled ... between knowledgeable and willing parties in an arm's length transaction''. 
Solvency assessment of an insurance company is based on future net values of assets and liabilities, and market-consistent valuation enables solvency assessments that takes dependence between future values of assets and liabilities into account.
Moreover, current regulatory frameworks emphasize risk measurement over a one-year period.
In particular, at any given time, the whole liability cash flow is taken into account in terms of the cash flow during the next one-year period and the market-consistent value at the end of the one-year period of the remaining liability cash flow. 
However, liability cash flows are typically not replicable by financial instruments. Therefore, the contribution to the liability value from the residual cash flow resulting from imperfect replication must be determined.

Given an aggregate liability cash flow of an insurance company, portfolios may be formed that generate cash flows with expected values matching that of the liability cash flow. 
The traditional actuarial practice of reserving provides an example of such a portfolio consisting of default-free bonds. In case of dependence between the liability cash flow and market values of financial instruments, more sophisticated replicating portfolios may be more suitable. 
However, the mismatch between the cash flow of such a portfolio and that of the original liability cash flow is typically substantial. The residual liability cash flow must be handled throughout the life of the liability cash flow by making sure that sufficient additional capital is available at all times. 
Capital providers, such as share holders, require compensation for providing buffer capital, which should be taken into account in the liability valuation. 
In particular, capital costs should be accounted for. 

In Solvency II, the so-called technical provisions correspond to the aggregate liability value and is defined as the sum of a best estimate, corresponding to a discounted actuarial fair value, and a so-called risk margin aimed at capturing capital costs. 
Unfortunately, the risk margin in the current regulatory framework lacks a proper definition and theoretical foundation, and different approximation formulas for this ill-defined object have been suggested. Criticism of the risk margin and suggestions for better notions of cost-of-capital margins or market-value margins are found in e.g.~\cite{Moehr-11}, \cite{Ohlsson-Lauzeningks-08}, \cite{Salzmann-Wuthrich-10}, \cite{Tsanakas-Wuthrich-Cerny-13} and \cite{Wuthrich-Embrechts-Tsanakas-11}. 

This paper addresses valuation of an aggregate liability cash flow of an insurance company, although the problem and our suggested solutions apply to liability valuation in other contexts as well. 
We present an approach that, in many aspects, is similar to current practice and has wide-ranging applicability. 
The approach we propose is inspired by \cite{Moehr-11} and \cite{Salzmann-Wuthrich-10}, where the cost-of-capital margin for valuing aggregate liability cash flows is analyzed.
The framework for liability cash flow valuation presented in \cite{Moehr-11} combines financial replication arguments with cost-of-capital considerations. Our proposed framework is on the one hand a generalization of that framework and with more attention paid to mathematical details. On the other hand, for the repeated one-period replication of the residual cash flow, we severely restrict the allowed replication instruments compared to \cite{Moehr-11}. Such a simplification of the problem allows us to derive much stronger results which in turn yields a framework that can be easily adopted, and it avoids many of the computational difficulties highlighted in \cite{Salzmann-Wuthrich-10} without sacrificing conceptual consistency. From a practical perspective, it leads to an approach to valuation of liability cash flows that does not rely heavily on subjective choices of joint dynamics for market prices of possible replication instruments.
The use of financial valuation principles in insurance is inevitable given the principle of market-consistent valuation of liabilities, in particular for liability cash flows with long durations and products with guarantees. 
For more on financial and actuarial valuation of insurance liabilities aimed at solvency assessment, see
 \cite{Wuthrich-Merz-13}.

Conditional monetary risk measures and utility functions are important basic building blocks in the approach to liability cash flow valuation considered here. Extensions from static, or one-period, risk measurement to dynamic risk measurement has been studied extensively for more than a decade following the seminal paper \cite{Artzner-Delbaen-Eber-Heath-99} on one-period risk measurement, 
see e.g.~\cite{Artzner-Delbaen-Eber-Heath-Ku-07}, 
\cite{Cheridito-Delbaen-Kupper-06}, 
\cite{Cvitanic-Karatzas-99},
\cite{Riedel-04},
\cite{Shapiro-12} and
\cite{Wang-99}.
Much of the analysis of dynamic risk measurement has focused on dynamic measurement of risk corresponding to a single cash flow, such as the cash flow of a derivative payoff, at a fixed future time. 
In \cite{Tsanakas-Wuthrich-Cerny-13} a market-value margin for valuation of insurance liabilities is presented based on multi-period mean-variance hedging. However, this valuation framework is not directly applicable to the problem we consider since the liability cash flow considered in \cite{Tsanakas-Wuthrich-Cerny-13} occur at a terminal time whereas we consider cash flows at all times up to a terminal time.
For the problem we consider, there is no natural way to roll cash flows forward and thereby reducing the dynamic risk measurement problem to a considerably simpler problem. The frameworks for dynamic risk measurement developed in \cite{Artzner-Delbaen-Eber-Heath-Ku-07} and \cite{Cheridito-Delbaen-Kupper-06} 
are well-suited to handle liability valuation problems of the type we consider. However, we do not want to restrict the liability cash flows to bounded stochastic processes. Moreover, the liability valuation problem we consider corresponds to repeated one-period replication rather than truly multi-period replication. Another important aspect is that we assume that the capital provider has limited liability and that causes the cost-of-capital margin to lack the convexity properties that are essential for the so-called risk-adjusted values analyzed in \cite{Artzner-Delbaen-Eber-Heath-Ku-07} and similarly for the dynamic risk measures analyzed in \cite{Cheridito-Delbaen-Kupper-06}.  
Risk measurement for multi-period income streams are studied in \cite{Pflug-Ruszczynski-03}
and \cite{Pflug-Ruszczynski-05}, where the dynamic risk measurement problem is formulated as a stochastic optimization problem. There, computational aspects of multi-period risk measurement are clarified and illustrated. Although our approach to liability valuation differ substantially from that in \cite{Pflug-Ruszczynski-03}
and \cite{Pflug-Ruszczynski-05}, computability is an essential feature.

This paper is organized as follows.

Section \ref{sec:COCM_intro} gives a nontechnical derivation of the cost-of-capital margin, as we believe that it should be defined, by economic arguments. 

Section \ref{sec:val_framework} presents a mathematical framework that allows the cost-of-capital margin to be defined rigorously, and establishes its fundamental properties. We also describe how the cost-of-capital margin is related to conditional monetary utility functions in the sense of 
\cite{Cheridito-Delbaen-Kupper-06}, showing that the cost-of-capital margin is conceptually consistent with dynamic monetary utility functions and risk-adjusted values as defined in \cite{Cheridito-Delbaen-Kupper-06} and \cite{Artzner-Delbaen-Eber-Heath-Ku-07}. There is however a major difference. The limited liability property of capital providers is an essential ingredient in our definition of the cost-of-capital margin and may cause violation of concavity/convexity properties. We define the cost-of-capital margin in terms of repeated one-period replication similar to \cite{Moehr-11} and allow capital requirements to be given in terms of conditional versions of nonconvex risk measures such as Value-at-Risk which is the current industry practice for insurance markets subject to the Solvency II regulation. Therefore, convexity properties are not assumed and not essential to us. Time-consistency is however an essential property. For the cost-of-capital margin, this property essentially follows  immediately from the definition.  
Towards the end of Section \ref{sec:val_framework} we consider a family of conditional risk measures that include commonly used risk measures such as Value-at-Risk and spectral risk measures, and we show that this family of risk measures are particularly useful for ensuring stronger properties and explicit formulas for the cost-of-capital margin when the liability cash flow is restricted to certain families of stochastic processes. It is well known, see e.g.~\cite{Cheridito-Stadje-09} and \cite{Shapiro-12}, that conditional or dynamic versions of Value-at-Risk and spectral risk measures are not time-consistent when applied to time periods of varying lengths. In our setting, only repeated conditional single-period risk measurement appears. Therefore, time-inconsistency of risk measurement over time periods of varying lengths does not cause problems for the time-consistency of the cost-of-capital margin. 

Section \ref{sec:CoCM_spec_models} considers specific models for the liability cash flow and the filtration representing the flow of information over time about the remaining cash flow until complete runoff of the liability. Specifically, we consider models of autoregressive type and Gaussian models.
We show that when combined with the general framework presented in Section \ref{sec:val_framework}, these models allow for explicit formulas and stronger results concerning the effects of properties of the chosen filtration. We believe that the explicit expressions presented here constitute candidates for standard formulas for cost-of-capital margin computation that may be adopted in improved future solvency regulation. 
 
Finally, Section \ref{life_example} presents a life-insurance example that illustrates features of the cost-of-capital margin and clarifies computational aspects.

\section{The cost-of-capital margin}\label{sec:COCM_intro}

In this section we derive the cost-of-capital margin without mathematical details, they are found in Section \ref{sec:val_framework}.
 
We consider time periods (years) $1,\dots,T$, corresponding time points $0,1,\dots,T$, and a filtered probability space $(\Omega,\calF,\filF,\P)$, where $\filF=(\calF_t)_{t=0}^{T}$ with $\{\emptyset, \Omega\}=\calF_0\subseteq \dots \subseteq \calF_{T}=\calF$. 

A liability cash flow corresponds to an $\filF$-adapted stochastic process $X^o=(X^o_t)_{t=1}^{T}$ interpreted as a cash flow from an aggregate insurance liability in runoff. Our aim is to give a precise meaning to the market-consistent value of the liability by taking capital costs into account, and provide results that allow this value to be computed. 

When the value of an insurance liability cash flow includes capital costs from capital requirements based on future values of both assets and liabilities, the liability value depends on the future values of all assets, including assets held for investment purpose only. In particular, two companies with identical liability cash flows would assign different market-consistent values to the two 
identical cash flows. This has undesired implications.
Instead, as is done in e.g.~\cite{Moehr-11} and prescribed by EIOPA, see \cite[Article 38]{Commission-del-reg-15}, we take the point of view that an aggregate liability cash flow should be valued by considering a hypothetical transfer of the liability to a separate entity, a so-called reference undertaking, whose assets have the sole purpose of matching the value of the liability as well as possible. 

We will give a meaning to the liability value by a particular two-stage valuation procedure: the first stage corresponds to choosing a replicating portfolio of traded financial instruments, the second stage corresponds to managing the residual cash flow from imperfect replication in the first stage. 

At time $0$, a portfolio is purchased with the aim of generating a cash flow replicating the cash flow $X^o$. This static replicating portfolio has a market price $\pi$ and generates the cash flow $X^s=(X^s_t)_{t=1}^T$. We use the wording ``static'' in order the emphasize that, for the purpose of valuation, it is a portfolio strategy that is fixed throughout the life of the liability cash flow. However, the replicating portfolio may be ``dynamic'' in the sense that its cash flow depends on events not known at time $0$.
This is completely in line with pricing a financial derivative in terms of the initial market price of a self-financing hedging strategy. 
If $X^o$ is independent of financial asset prices, then the canonical example is a portfolio of zero-coupon bonds generating the cash flow $X^s=\E{X^o}=(\E{X^o_t})_{t=1}^T$. 

The value of the original liability is defined as the sum of the market price $\pi$ of the replicating portfolio and the value $V_0(X)$ of the residual cash flow $X:=X^o-X^s$ from repeated one-period replication using only cash funds provided by a capital provider with limited liability requiring compensation for capital costs. That is, the cash flow $X$ will be re-valued at all times $1,\dots,T$. 
We call the value $V_0(X)$ of the residual cash flow the cost-of-capital margin. Note that we below refer also to $V_t(X)$ as the cost-of-capital margin for all $t = 0,\dots,T$. We emphasize that repeated one-period replication is done with cash only. Allowing for repeated one-period replication using a mix of assets such as bonds with short time to maturity inevitably makes the value of the liability cash flow depend on subjective views on the development over time of spot rates over different time horizons.

Next we will present economic arguments that lead to a recursion defining $V_t(X)$ in terms of $X_{t+1}$ and $V_{t+1}(X)$, capital requirements and the acceptability condition of the capital provider.
At time $t$, the insurance company is required to hold the capital $R_t(-X_{t+1}-V_{t+1}(X))$, where $R_t$ is a conditional monetary risk measure, see Definition \ref{def:dynrisk}, quantifying the risk from liability cash flow during year $t+1$ and the value at time $t+1$ of the remaining residual cash flow $(X_{t+2},\dots,X_T)$.
The capital provider is asked to provide the capital 
\begin{align}\label{eq:SCR}
C_t:=R_t(-X_{t+1}-V_{t+1}(X))-V_t(X).
\end{align}
The amount $C_t$ is the difference between the required capital 
$R_t(-X_{t+1}-V_{t+1}(X))$ and the value $V_t(X)$ of the residual liability cash flow at time $t$. 
If the capital provider accepts providing $C_t$ at time $t$, then at time $t+1$ the capital $R_t(-X_{t+1}-V_{t+1}(X))-X_{t+1}$ is available. If this amount exceeds the value $V_{t+1}(X)$ of the liability at time $t+1$, then the capital provider collects the excess capital as a compensation for providing the buffer capital $C_t$ at time $t$. 
Moreover, the capital provider has limited liability: if 
\begin{align*}
R_t(-X_{t+1}-V_{t+1}(X))-X_{t+1}-V_{t+1}(X)<0,
\end{align*}
then the capital provider has no obligation to provide further capital to offset the deficit. 

The capital provider's acceptability condition at time $t$ is expressed in terms of a conditional monetary utility function $U_t$, see Definition \ref{def:dynutil}, and a value $\eta_t>0$ quantifying the size of the intended compensation to the capital provider for making capital available:
\begin{align}\label{eq:ac_U}
U_t\big(\big(R_t(-X_{t+1}-V_{t+1}(X))-X_{t+1}-V_{t+1}(X)\big)_+\big)\geq (1+\eta_t)C_t,
\end{align}
where $x_+=\max(x,0)$. 
$U_t$ may be chosen as the conditional expectation although alternatives that take the risk aversion of the capital provider into account may be more appropriate.

Combining \eqref{eq:SCR} and \eqref{eq:ac_U} now gives, with $Y_{t+1}:=X_{t+1}+V_{t+1}(X)$,
\begin{align*}
V_t(X)\geq R_t(-Y_{t+1})-\frac{1}{1+\eta_t}U_t\big(\big(R_t(-Y_{t+1})-Y_{t+1}\big)_+\big).
\end{align*}
If the inequality above is strict, then the capital provider obtains a better-than-required investment opportunity at the expense of policy holders who are obliged to pay higher-than-needed premiums.
Therefore, we define the value of the cash flow as the smallest value for which the capital provider finds the investment opportunity acceptable. That is, we replace the inequality above by an equality:
\begin{align}\label{eq:recursive_valuation}
V_t(X):=R_t(-Y_{t+1})-\frac{1}{1+\eta_t}U_t\big(\big(R_t(-Y_{t+1})-Y_{t+1}\big)_+\big).
\end{align}
Recall that the market-consistent value we assign to the original cash flow is the sum of the market price $\pi$ of a replicating portfolio set up at time $0$ and the cost-of-capital margin $V_0(X)$.
Notice that if perfect initial replication is possible, then $X=0$ as there is no residual cash flow, and consequently $V_0(X)=0$ since no capital funds for repeated one-period replication are needed.


\section{The valuation framework}\label{sec:val_framework}

We consider time periods $1,\dots,T$, corresponding time points $0,1,\dots,T$, and a filtered probability space $(\Omega,\calF,\filF,\P)$, where $\filF=(\calF_t)_{t=0}^{T}$ with $\{\emptyset, \Omega\}=\calF_0\subseteq \dots \subseteq \calF_{T}=\calF$. 
Let $L^0(\calF_t):=L^0(\Omega,\calF_t,\P)$ denote the vector space of all real-valued $\calF_t$-measurable random variables, and let $L^0_+(\calF_t)$ be the subset of $L^0(\calF_t)$ of random variables taking values in $[0,\infty)$.
$L^0(\calF_0)$ is simply the set of constants, i.e.~$\R$.
For $p\in (0,\infty)$, let $L^p(\calF_t):=\{Y\in L^0(\calF_t):\E{|Y|^p}<\infty\}$ and let $L^p_+(\calF_t)$ denotes the  subset of $L^p(\calF_t)$ of random variables taking values in $[0,\infty)$.
Finally, $L^{\infty}(\calF_t)$ consists of the essentially bounded $\calF_t$-measurable random variables: $Y\in L^0(\calF_t)$ such that 
\begin{align*}
\inf\{r>0:\P(\omega\in\Omega:|Y(\omega)|>r)=0\}<\infty.
\end{align*}
We say that two random variables are equal if they coincide $\P$-almost surely (a.s.). All equalities and inequalities between random variables are interpreted in the $\P$-a.s. sense.  

In order to determine the cost-of-capital margin we consider conditional monetary risk measures $R_t$ and conditional monetary utility functions $U_t$. We express values and cash flows via a num\'eraire which we take to be a money market account that pays no interest, where money can be deposited and later withdrawn. We have no need for and do not assume risk-free borrowing. Choosing the num\'eraire to be a money market account paying stochastic interest rates would not pose mathematical difficulties but would force us to pay more attention to the interpretation of the cash flows. 

By a dynamic monetary risk measure $R:=(R_t)_{t=0}^{T-1}$ quantifying one-period capital requirements we mean the following:

\begin{definition}\label{def:dynrisk}
For $p\in [0,\infty]$, a dynamic monetary risk measure $(R_t)_{t=0}^{T-1}$ is a sequence of mappings $R_t:L^p(\calF_{t+1})\to L^p(\calF_t)$ satisfying
\begin{align}
& \textrm{if } \lambda\in L^p(\calF_t) \textrm{ and } Y\in L^p(\calF_{t+1}), \textrm{ then } 
R_t(Y+\lambda)=R_t(Y)-\lambda,  \label{eq:ti_r}\\
& \textrm{if } Y,\widetilde{Y}\in L^p(\calF_{t+1}) \textrm{ and } Y\leq \widetilde{Y}, \textrm{ then } 
R_t(Y)\geq R_t(\widetilde{Y}),  \label{eq:mo_r}\\
&\textrm{if } c\in L^p_+(\calF_t) \textrm{ and } Y\in L^p(\calF_{t+1}), \textrm{ then } 
R_t(cY)=cR_t(Y). \label{eq:ph_r}
\end{align}
\end{definition}
We refer to the properties \eqref{eq:ti_r}-\eqref{eq:ph_r} as translation invariance, monotonicity and positive homogeneity, respectively.

By a dynamic monetary utility function $U:=(U_t)_{t=0}^{T-1}$ quantifying one-period acceptability for capital providers we mean the following:
  
\begin{definition}\label{def:dynutil}
For $p\in [0,\infty]$, a dynamic monetary utility function $(U_t)_{t=0}^{T-1}$ is a sequence of mappings $U_t:L^p(\calF_{t+1})\to L^p(\calF_t)$ satisfying
\begin{align}
& \textrm{if } \lambda\in L^p(\calF_t) \textrm{ and } Y\in L^p(\calF_{t+1}), \textrm{ then } 
U_t(Y+\lambda)=U_t(Y)+\lambda, \label{eq:ti_u}\\
& \textrm{if } Y,\widetilde{Y}\in L^p(\calF_{t+1}) \textrm{ and } Y\leq \widetilde{Y}, \textrm{ then } 
U_t(Y)\leq U_t(\widetilde{Y}), \label{eq:mo_u}\\
&\textrm{if } c\in L^p_+(\calF_t) \textrm{ and } Y\in L^p(\calF_{t+1}), \textrm{ then } 
U_t(cY)=cU_t(Y). \label{eq:ph_u}
\end{align}
\end{definition}
We refer also to the properties \eqref{eq:ti_u}-\eqref{eq:ph_u} as translation invariance, monotonicity and positive homogeneity, respectively.

The following proposition provides the basis for defining the cost-of-capital margin in \eqref{eq:recursive_valuation} rigorously.

\begin{proposition}\label{lem:basic_lem}
Fix $p\in [0,\infty]$. Let $(R_t)_{t=0}^{T-1}$ and $(U_t)_{t=0}^{T-1}$ be given by Definitions \ref{def:dynrisk} and \ref{def:dynutil}, respectively, and let $0<\eta_t\in L^0(\calF_t)$. \\
(i) 
\begin{align}\label{eq:vy}
W_t(Y):=R_t(-Y)-\frac{1}{1+\eta_t}U_t\big(\big(R_t(-Y)-Y\big)_+\big)
\end{align}
is a mapping from $L^p(\calF_{t+1})$ to $L^p(\calF_{t})$ having the properties
\begin{align}
& \textrm{if } \lambda\in L^p(\calF_t) \textrm{ and } Y\in L^p(\calF_{t+1}), \textrm{ then } 
W_t(Y+\lambda)=W_t(Y)+\lambda, \label{eq:ti_w}\\
& \textrm{if } Y,\widetilde{Y}\in L^p(\calF_{t+1}) \textrm{ and } Y\leq \widetilde{Y}, \textrm{ then } 
W_t(Y)\leq W_t(\widetilde{Y}), \label{eq:mo_w}\\
&\textrm{if } c\in L^p_+(\calF_t) \textrm{ and } Y\in L^p(\calF_{t+1}), \textrm{ then } 
W_t(cY)=cW_t(Y). \label{eq:ph_w}
\end{align}
(ii)
Let $(X_t)_{t=1}^T$ be an $\filF$-adapted cash flow with $X_t\in L^p(\calF_t)$ for every $t$. 
The cost-of-capital margin $V_t(X)$ in \eqref{eq:recursive_valuation} satisfies 
\begin{align}\label{eq:VtWs_rep}
V_t(X)=W_t\circ\dots\circ W_{T-1}(X_{t+1}+\dots+X_T),
\end{align}
where $W_t\circ\dots\circ W_{T-1}$ denotes the composition of mappings 
$W_t,\dots, W_{T-1}$.
\end{proposition}

\begin{proof}[Proof of Proposition \ref{lem:basic_lem}]
Since $(1+\eta_t)^{-1}$ is $\calF_t$-measurable and takes values in $(0,1)$, it follows directly from the definitions of $R_t$ and $U_t$ that $W_t$ is a mapping from $L^p(\calF_{t+1})$ to $L^p(\calF_t)$.
The properties \eqref{eq:ti_w} and \eqref{eq:ph_w} for $W_t$ follow immediately from the corresponding properties of $R_t$ and $U_t$. It remains to verify property \eqref{eq:mo_w} for $W_t$. Take $Y_{t+1}\leq \widetilde{Y}_{t+1}$ in $L^p(\calF_{t+1})$. Then
\begin{align*}
(R_t(-\widetilde{Y}_{t+1})-\widetilde{Y}_{t+1})_+
&\leq (R_t(-Y_{t+1})-\widetilde{Y}_{t+1})_++R_t(-\widetilde{Y}_{t+1})-R_t(-Y_{t+1})\\
&\leq (R_t(-Y_{t+1})-Y_{t+1})_++R_t(-\widetilde{Y}_{t+1})-R_t(-Y_{t+1}).
\end{align*}
Since $R_t(-\widetilde{Y}_{t+1})-R_t(-Y_{t+1})\in L^p_+(\calF_t)$, \eqref{eq:ti_u} and \eqref{eq:mo_u} together imply that
\begin{align*}
&U_t((R_t(-\widetilde{Y}_{t+1})-\widetilde{Y}_{t+1})_+)-U_t((R_t(-Y_{t+1})-Y_{t+1})_+)\\
&\quad\leq R_t(-\widetilde{Y}_{t+1})-R_t(-Y_{t+1})
\end{align*}
which further implies
\begin{align*}
W_t(\widetilde{Y}_{t+1})-W_t(Y_{t+1})\geq (R_t(-\widetilde{Y}_{t+1})-R_t(-Y_{t+1}))\frac{\eta_t}{1+\eta_t}\geq 0
\end{align*}
and verifies the property \eqref{eq:mo_w}. 
Finally, we verify the representation \eqref{eq:VtWs_rep} of $V_t(X)$ in terms of $W_t,\dots,W_{T-1}$ and $X_{t+1}+\dots+X_T$. If $X_s\in L^p(\calF_s)$ for $s=t+1,\dots,T$, then $X_{t+1}+\dots+X_T\in L^p(\calF_T)$ and the right-hand side in \eqref{eq:VtWs_rep} is well-defined. Repeated application of \eqref{eq:ti_w} now verifies the representation \eqref{eq:VtWs_rep}. 
\end{proof}

\begin{definition}\label{def:COCM}
Fix $p\in [0,\infty]$. Let $(R_t)_{t=0}^{T-1}$ and $(U_t)_{t=0}^{T-1}$ be given by Definitions \ref{def:dynrisk} and \ref{def:dynutil}, respectively, and let $(W_t)_{t=0}^{T-1}$ be given by \eqref{eq:vy}.
Let $(X_t)_{t=1}^T$ be an $\filF$-adapted cash flow with $X_t\in L^p(\calF_t)$ for every $t$. 
We define the cost-of-capital margins $V_t(X)$ as 
\begin{align}\label{eq:VtWs_rep_def}
V_t(X):=W_t\circ\dots\circ W_{T-1}(X_{t+1}+\dots+X_T), \quad t=0,\dots,T-1,
\end{align}
and $V_T(X):=0$.
\end{definition}

Notice that we may express \eqref{eq:VtWs_rep_def} as 
$V_t(X):=W_t(X_{t+1}+V_{t+1}(X))=W_t(Y_{t+1})$.

\begin{proposition}\label{prop:basic_prop}
Let $X,\widetilde{X}$ be $\filF$-adapted cash flows with $X_t,\widetilde{X}_t\in L^p(\calF_t)$ for every $t$. \\
(i) Let $a\in L^p_+(\calF_t)$, let $b$ be a $T$-dimensional vector with components in $L^p(\calF_t)$, and let $X_u\leq \widetilde{X}_u$ for each $u$. Then, for every $t<T$, 
\begin{align*}
V_t(aX)=aV_t(X), \quad V_t(X+b)=V_t(X)+\sum_{u=t+1}^Tb_t, \quad V_t(X)\leq V_t(\widetilde{X}).
\end{align*} 
(ii) The cost-of-capital margins are time consistent in the sense that for 
every pair of times $(s,t)$ with $s\leq t$, the two conditions $(X_u)_{u=1}^{t}=(\widetilde{X}_u)_{u=1}^{t}$ and $V_t(X)\leq V_t(\widetilde{X})$ together imply $V_s(X)\leq V_s(\widetilde{X})$.
\end{proposition}

\begin{proof}[Proof of Proposition \ref{prop:basic_prop}]
(i) The properties follow immediately from \eqref{eq:VtWs_rep_def}.
(ii) In proving time consistency it is sufficient to take $s=t-1$. By \eqref{eq:mo_w},
\begin{align*}
V_{t-1}(X)&=W_{t-1}(X_t+V_t(X))\\
&=W_{t-1}(\widetilde{X}_t+V_t(X))\\
&\leq W_{t-1}(\widetilde{X}_t+V_t(\widetilde{X}))\\
&=V_{t-1}(\widetilde{X}).
\end{align*} 
\end{proof}

The dynamic version of Value-at-Risk presented in the example below is an example of a dynamic monetary risk measure $(R_t)_{t=0}^{T-1}$ with $R_t:L^0(\calF_{t+1})\to L^0(\calF_t)$.
In Section \ref{sec:nice_mappings} further examples of dynamic monetary risk measures and utility functions are presented and their properties are investigated for use in Section \ref{sec:CoCM_spec_models} together with specific models for the liability cash flows.

\begin{example}
In the static or one-period setting, Value-at-Risk at time $0$ at level $u\in (0,1)$ of a value $Z\in L^0(\calF_{1})$ is defined as 
\begin{align*}
\VaR_u(Z)&:=\min\{m\in\R:\P(m+Z<0)\leq u\}\\
&=\min\{m\in\R:\P(-Z\leq m)\geq 1-u\}\\
&=\min\{m\in\R:Q_{-Z}((-\infty,m])\geq 1-u)\}\\
&=:F_{-Z}^{-1}(1-u),
\end{align*}
where $Q_{-Z}$ denotes the distribution of $-Z$, and $F_{-Z}(m)=Q_{-Z}((-\infty,m])$.
The natural dynamic version of Value-at-Risk at time $t$ at level $u$ of a value $Z\in L^0(\calF_{t+1})$ is
\begin{align*}
\VaR_{t,u}(Z):=\ess\,\inf\{m\in L^0(\calF_t):\P(m+Z<0\mid\calF_t)\leq u\},
\end{align*}
where ``$\ess\,\inf$'' denotes the greatest lower bound of a family of random variables (with respect to $\P$-almost sure inequality). Alternatively, we may define $\VaR_{t,u}(Z)$ in terms of a conditional distribution $Q_{t,-Z}$ of $-Z$ with respect to $\calF_t$:
for each $\omega\in\Omega$, $Q_{t,-Z}(\omega,\cdot)$ is a probability measure on the Borel subsets of $\R$, and for each Borel set $A\subset \R$, $Q_{t,-Z}(\cdot,A)$ is a version of $\P(-Z\in A\mid\calF_t)$.
Define $Z'\in L^0(\calF_t)$ by 
\begin{align*}
Z'(\omega)&:=\min\{m\in\R:Q_{t,-Z}(\omega,(-\infty,m])\geq 1-u\}
=:F_{t,-Z}^{-1}(\omega,1-u)
\end{align*}
and notice that $\VaR_{t,u}(Z)=Z'$ $\P$-almost surely. Notice that $(\VaR_{t,u})_{t=0}^{T-1}$ satisfies the properties in Definition \ref{def:dynrisk} for $p=0$.
\end{example}

\subsection{Model-invariant bounds}\label{sec:cocm_bounds}

Consider a dynamic risk measure $(R_t)_{t=0}^{T-1}$ and a dynamic monetary utility function $(U_t)_{t=0}^{T-1}$, and $(W_t)_{t=0}^{T-1}$ given by \eqref{eq:vy}. 
Consider also an $\filF$-adapted cash flow $(X_t)_{t=1}^T$ with $X_t\in L^p(\calF_t)$ for every $t$.
With $Y_{t+1}:=X_{t+1}+V_{t+1}(X)$,
\begin{align*}
V_t(X)&=R_t(-Y_{t+1})-\frac{1}{1+\eta_t}U_t((R_t(-Y_{t+1})-Y_{t+1})_+)\\
&\leq R_t(-Y_{t+1})\\
&=R_t(-X_{t+1}-V_{t+1}(X)).
\end{align*}
Repeated application of this inequality together with \eqref{eq:ti_r} and \eqref{eq:mo_r} gives the upper bound
\begin{align*}
V_{t}(X)&\leq R_t(-X_{t+1}-R_{t+1}(-X_{t+2}-\dots-R_{T-1}(-X_T)\dots))\\
&=R_t(-R_{t+1}(\dots-R_{T-1}(-X_{t+1}-\dots-X_T)\dots)).
\end{align*}
Further assumptions clearly enable sharper bounds. Suppose that $p\geq 1$ and take, for every $t$, $U_t$ to be the conditional expectation given $\calF_t$. 
It follows from Jensen's inequality for conditional expectations that the conditional expectation is well defined as a mapping $L^p(\calF_{u})\to L^p(\calF_t)$ for $u>t$.
In particular,
\begin{align}
V_t(X)&=R_t(-Y_{t+1})-\frac{1}{1+\eta_t}\E{\big(R_t(-Y_{t+1})-Y_{t+1}\big)_+\mid\calF_t} \nonumber\\
&\leq R_t(-Y_{t+1})-\frac{1}{1+\eta_t}\E{R_t(-Y_{t+1})-Y_{t+1}\mid\calF_t} \nonumber\\
&=\frac{1}{1+\eta_t}\Big(\eta_tR_t(-Y_{t+1})+\E{X_{t+1}\mid\calF_t}+\E{V_{t+1}(X)\mid\calF_t}\Big).\label{eq:Vtupper_bound}
\end{align}
Applying this inequality repeatedly together with the tower property of conditional expectation yields
\begin{align*}
V_t(X)\leq 
\sum_{s=t}^{T-1}\E{\frac{\eta_sR_s(-Y_{s+1})}{\prod_{u=t}^s(1+\eta_u)}\mid\calF_t}
+\sum_{s=t}^{T-1}\E{\frac{X_{s+1}}{\prod_{u=t}^s(1+\eta_u)}\mid\calF_t}.
\end{align*}
Notice that if $\eta_t=\eta_0$ for all $t$, then 
\begin{align*}
V_0(X)&\leq 
\eta_0\sum_{t=0}^{T-1}\frac{\E{R_t(-Y_{t+1})}}{(1+\eta_0)^{t+1}}
+\sum_{t=1}^{T}\frac{\E{X_t}}{(1+\eta_0)^{t}}.
\end{align*}
Notice that if, further, the static replicating portfolio is chosen at time $0$ such that $\E{X^s}=\E{X^o}$, then the residual cash flow $X:=X^o-\E{X^o}$ has zero mean. In particular, then the second sum in the above upper bound vanishes, i.e.
\begin{align}\label{eq:V0upper_bound}
V_0(X)\leq \eta_0\sum_{t=0}^{T-1}\frac{\E{R_t(-Y_{t+1})}}{(1+\eta_0)^{t+1}}.
\end{align}
Notice that the upper bound for $V_t(X)$ in \eqref{eq:Vtupper_bound} can be rewritten as
\begin{align*}
V_t(X)\leq \eta_t(R_t(-Y_{t+1})-V_t(X))+\E{X_{t+1}\mid\calF_t}+\E{V_{t+1}(X)\mid\calF_t}.
\end{align*}
Repeated application of this inequality and use of the tower property of conditional expectation yields
\begin{align*}
V_0(X)\leq\sum_{t=0}^{T-1}\E{\eta_t\Big(R_t(-Y_{t+1})-V_t(X)\Big)}+\sum_{t=1}^T\E{X_t}.
\end{align*}
Hence, if further $\eta_t=\eta_0$ for all $t$ and $\E{X}=0$, then 
\begin{align}\label{eq:V0upper_bound_alt}
V_0(X)\leq\eta_0\sum_{t=0}^{T-1}\E{R_t(-Y_{t+1})-V_t(X)}.
\end{align}
Notice the difference between the two upper bounds 
\eqref{eq:V0upper_bound} and \eqref{eq:V0upper_bound_alt}:
the former is formulated in terms of expected future capital requirement whereas 
the latter is formulated in terms of expected future buffer capital provided by capital providers.

We end the discussion of model-invariant bounds for the cost-of-capital margin with a comment on the Solvency II risk margin. In \cite[Article 37]{Commission-del-reg-15} it is stated that the risk margin should be computed as 
\begin{align*}
\CoC \sum_{t\geq 0}\frac{\SCR(t)}{(1+r(t+1))^{t+1}},
\end{align*}
where $\CoC:=0.06$, $r(t+1)$ denotes the basic risk-free interest rate for the maturity of $t+1$ years, and $\SCR(t)$ denotes the Solvency Capital Requirement after $t$ years. In our setting, $\SCR(t)=0$ for $t\geq T$ since there is no liability cash flow beyond that time. One may criticize several aspects of the Solvency II risk margin. First, for $t\geq 1$, $\SCR(t)$ is a random variable as seen from time $0$. Secondly, $\SCR(t)$ does not take capital costs into account, and the discounting in the computation of $\SCR(t)$ and in the computation of the risk margin are not conceptually consistent.  

The upper bound in \eqref{eq:V0upper_bound_alt} is somewhat similar to the formula for the Solvency II risk margin. 
For $p\geq 1$, take $(R_t)_{t=0}^{T-1}$ to be a dynamic risk measure in the sense of Definition \ref{def:dynrisk}, and take the $U_t$ to be conditional expectations $\E{\cdot \mid \calF_t}$.
If we define the SCR-like quantity 
\begin{align*}
\widetilde{\SCR}_t(X^o)&:=R_t\Big(\sum_{u=t+1}^T\E{X^o_u\mid\calF_t}+V_t(X^o-\E{X^o\mid\calF_t})\\
&\quad\quad\quad-X^o_{t+1}-\sum_{u=t+2}^T\E{X^o_u\mid\calF_{t+1}}-V_{t+1}(X^o-\E{X^o\mid\calF_{t+1}})\Big),
\end{align*}
then, using the translation invariance of $V_t$ and $V_{t+1}$ in Proposition \ref{prop:basic_prop}, it can be easily shown that 
$\widetilde{\SCR}_t(X^o)=R_t(-Y_{t+1})-V_t(X)$,
where, as before, $X:=X^o-\E{X^o\mid\calF_0}$ and $Y_{t+1}:=X_{t+1}+V_{t+1}(X)$. 
In particular, \eqref{eq:V0upper_bound_alt} can be rephrased as 
\begin{align*}
V_0(X)\leq\eta_0\sum_{t=0}^{T-1}\E{\widetilde{\SCR}_t(X^o)}.
\end{align*} 
For an in-depth comparison between a conceptually consistent notion of cost-of-capital margin and the Solvency II risk margin, see Section 5 in \cite{Moehr-11}.

\subsection{Dynamic monetary utility functions}

The notions of conditional and dynamic monetary utility functions in 
\cite{Cheridito-Delbaen-Kupper-06} and risk-adjusted values in 
\cite{Artzner-Delbaen-Eber-Heath-Ku-07}
are closely connected to the cost-of-capital margin considered here. In \cite{Cheridito-Delbaen-Kupper-06} and \cite{Artzner-Delbaen-Eber-Heath-Ku-07}, cumulative cash flows and value processes are considered whereas we consider incremental cash flows and liability processes corresponding to liability values for future cash flows. We will now proceed to establish the connection between $V_t(\cdot)$ (and $W_t(\cdot)$) and the work of \cite{Cheridito-Delbaen-Kupper-06} and \cite{Artzner-Delbaen-Eber-Heath-Ku-07}.

For $p\in [0,\infty]$, let $\calR^p_{1,T}$ denote the space of all $\filF$-adapted stochastic processes $(Y_t)_{t=1}^T$ with $Y_t\in L^p(\calF_t)$ for every $t$.
For $1\leq u\leq v\leq T$, define the projection $\pi_{u,v}:\calR^p_{1,T}\to\calR^p_{1,T}$ by
\begin{align*}
\pi_{u,v}(Y)_t:=1_{u\leq t}Y_{t\wedge v}, \quad t\in \{1,\dots,T\},
\end{align*} 
and $\calR^p_{u,v}:=\pi_{u,v}\calR^p_{1,T}$.
Let $(W_t)_{t=0}^{T-1}$ be as in Proposition \ref{lem:basic_lem}.
For all $t$, let $\phi_{t,T}:\calR^p_{t,T}\to L^p(\calF_t)$ be given by 
\begin{align}\label{eq:CDK}
\phi_{t,T}(Y)&:=-W_t\circ \dots \circ W_{T-1}(-Y_T).
\end{align}
In \cite{Cheridito-Delbaen-Kupper-06} and \cite{Artzner-Delbaen-Eber-Heath-Ku-07} the elements in $\calR^p_{1,T}$ are interpreted as cumulative rather than incremental cash flows.
For an incremental cash flow $X\in \calR^p_{1,T}$, let 
$X^c\in\calR^p_{1,T}$ be the cumulative cash flow given by $X^c_t:=\sum_{s=1}^t X_s$, and notice that 
\begin{align*}
\phi_{t,T}(X^c)&=X^c_t-W_t\circ \dots \circ W_{T-1}(-(X^c_T-X^c_t))\\
&=\sum_{s=1}^tX_s-V_t(-X),
\end{align*}
and similarly, $V_t(X)=-X^c_t-\phi_{t,T}(-X^c)$.

We now verify that $\phi_{t,T}$ in \eqref{eq:CDK} is a conditional monetary utility function in the sense of Definition 3.1 in \cite{Cheridito-Delbaen-Kupper-06}, excluding the concavity axiom, and that $(\phi_{t,T})_{t=1}^T$ is time-consistent in the sense of Definition 4.2 in \cite{Cheridito-Delbaen-Kupper-06}. 

\begin{proposition}\label{prop:CDK}
The mappings $\phi_{t,T}$ in \eqref{eq:CDK} are conditional monetary utility functions in the sense
\begin{align*}
&\phi_{t,T}(0)=0,\\
& \phi_{t,T}(Y)\leq \phi_{t,T}(\widetilde{Y}) \text{ for all } Y,\widetilde{Y}\in\calR^p_{t,T} \text{ such that }
Y\leq \widetilde{Y},\\
& \phi_{t,T}(Y+m1_{[t,T]})=\phi_{t,T}(Y)+m \text{ for all } Y\in\calR^p_{t,T} \text{ and } m\in L^p(\calF_t),
\end{align*} 
and $(\phi_{t,T})_{t=1}^T$ is time-consistent in the sense 
\begin{align}\label{eq:time_consistency_delbaen}
\phi_{t,T}(Y)=\phi_{t,T}(Y1_{[t,u)}+\phi_{u,T}(Y)1_{[u,T]})
\end{align} 
for every $t\leq u \leq T$ and all $X\in\calR^p_{t,T}$.
\end{proposition}

\begin{proof}[Proof of Proposition \ref{prop:CDK}]
Since $W_s(0)=0$ for $s =0, ..., T-1$,
\begin{align*}
\phi_{t,T}(0)=-W_t\circ \dots \circ W_{T-1}(-0)=0.
\end{align*} 
Due to the fact that $W_s(0)=0$ for $s =0, \dots, T-1$. 
For $Y,\widetilde{Y}\in\calR^p_{t,T}$ such that $Y\leq \widetilde{Y}$,
\begin{align*}
&\phi_{t,T}(\widetilde{Y})-\phi_{t,T}(Y)=-W_t\circ \dots \circ W_{T-1}(-\widetilde{Y}_T)+W_t\circ \dots \circ W_{T-1}(-Y_T)
\end{align*}
Noting that $-Y_T\geq-\widetilde{Y}_T$ and using, repeatedly, the monotonicity property of $W_s$ for $s =0, \dots,T-1$ we arrive at $\phi_{t,T}(\widetilde{Y})-\phi_{t,T}(Y)\geq 0$.
Finally we note that
\begin{align*}
\phi_{t,T}(Y+m1_{[t,T]})&=-W_t\circ \dots \circ W_{T-1}(-(Y_T+m))\\
&=m-W_t\circ \dots \circ W_{T-1}(-Y_T)\\
&= \phi_{t,T}(Y)+m,
\end{align*}
where the second equality follows from the translation invariance of $W_s$ for $s=t, \dots,T-1$. Time consistency in the sense of \eqref{eq:time_consistency_delbaen} follows almost directly from the definition. 
\begin{align*}
&\phi_{t,T}(Y1_{[t,u)} +\phi_{u,T}(Y)1_{[u,T]})\\
&\quad=  -W_t\circ \dots \circ W_{T-1}(-(Y1_{[t,u)} +\phi_{u,T}(Y)1_{[u,T]})_T)\\
&\quad=  -W_t\circ \dots \circ W_{T-1}(-\phi_{u,T}(Y))\\
&\quad=  -W_t\circ \dots \circ W_{u-1}(-\phi_{u,T}(Y)),
\end{align*}
where the last equality is due to the translation invariance of $W_s$ for $s =t,\dots,T-1$ combined with the fact that $-\phi_{u,T}(Y)$ is $\calF_s$-measurable for $s\geq u$ and noting that $W_s(0)=0$ for $s =u,\dots,T-1$. From the definition \eqref{eq:CDK} of $\phi_{u,T}$ we get
\begin{align*}
& -W_t\circ \dots \circ W_{u-1}(-( -W_u\circ \dots \circ W_{T-1}(-Y_T)))\\
&\quad= -W_t\circ \dots \circ W_{T-1}(-Y_T)\\
&\quad=\phi_{t,T}(Y).
\end{align*}
\end{proof}

\subsection{Risk measures and utility functions based on conditional quantiles}\label{sec:nice_mappings}

For $Y\in L^p(\calF_{t+1})$, write $Q_{t,Y}$ for its conditional distribution given $\calF_t$: $Q_{t,Y}(\omega,\cdot)$ is a probability measure on the Borel subsets of $\R$, and $Q_{t,Y}(\cdot,A)$ is a version of $\P(Y\in A\mid\calF_t)$. We may write the conditional distribution and quantile functions of $Y$ given $\calF_t$, respectively, as
\begin{align*}
F_{t,Y}(\omega,y)&:=Q_{t,Y}(\omega,(-\infty,y]),\\
F_{t,Y}^{-1}(\omega,u)&:=\min\{y\in\R:F_{t,Y}(\omega,y)\geq u\}.
\end{align*} 
For a probability measures $M^R$ and $M^U$ on $(0,1)$, define, for $Y\in L^p(\calF_{t+1})$,
\begin{align}
R_t(Y)&:=\int_0^1F_{t,-Y}^{-1}(u)dM^R(u), \label{eq:quantileR}\\
U_t(Y)&:=\int_0^1F_{t,Y}^{-1}(u)dM^U(u).\label{eq:quantileU}
\end{align}

\begin{proposition}\label{prop:quantileRandUandW}
Suppose there exist $u_0\in (0,1)$ and $\overline{m}\in (0,\infty)$ such that, for $k=R,U$, 
\begin{align*}
\max\Big(M^k((u,v),M^k((1-v,1-u))\Big)\leq \overline{m}(v-u)\quad\text{for all } 0<u<v<u_0.
\end{align*}
Fix $p\in [1,\infty]$. \\
(i) $R_t$ in \eqref{eq:quantileR} and $U_t$ in \eqref{eq:quantileU} are well-defined as mappings from $L^p(\calF_{t+1})$ to $L^p(\calF_{t})$ and satisfy 
\eqref{eq:ti_r}-\eqref{eq:ph_r} and \eqref{eq:ti_u}-\eqref{eq:ph_u}, respectively. \\
(ii)
If $Y\in L^p(\calF_{t+1})$ and, for any Borel set $A\subset \R$ 
\begin{align*}
\P(Y\in A\mid\calF_t)=\P(Y^{(1)}+Y^{(2)}Y^{(3)}\in A\mid\calF_t), 
\end{align*}
where $Y^{(1)}\in L^p(\calF_t)$, $0<Y^{(2)}\in L^0(\calF_0)$, and $Y^{(3)}\in L^p(\calF_{t+1})$ is independent of $\calF_t$, then 
\begin{align*}
R_t(Y)&=Y^{(1)}+Y^{(2)}R_t(Y^{(3)}),\\
U_t(Y)&=Y^{(1)}+Y^{(2)}U_t(Y^{(3)}),
\end{align*}
where $R_t(Y^{(3)}),U_t(Y^{(3)})\in L^0(\calF_0)$. 
Moreover, $R_t(Y^{(3)})=R_0(\widetilde{Y}^{(3)})$ and $U_t(Y^{(3)})=U_0(\widetilde{Y}^{(3)})$ for $\widetilde{Y}^{(3)}\in L^p(\calF_1)$ with $Y^{(3)}$ and $\widetilde{Y}^{(3)}$ equally distributed.
For $W_t$ in \eqref{eq:vy}, 
\begin{align*}
W_t(Y)=Y^{(1)}+Y^{(2)}W_t(Y^{(3)}), 
\end{align*}
and if 
$0<\eta_t\in L^0(\calF_0)$, then $W_t(Y^{(3)})\in L^0(\calF_0)$. Further, if $\eta_t=\eta_0$, then $W_t(Y^{(3)})=W_0(\widetilde{Y}^{(3)})$ for $\widetilde{Y}^{(3)}\in L^p(\calF_1)$ with $Y^{(3)}$ and $\widetilde{Y}^{(3)}$ equally distributed.
\end{proposition}

\begin{proof}[Proof of Proposition \ref{prop:quantileRandUandW}]
(i) It is sufficient to prove the statement (i) for $R_t$ and $U_t$ only for $R_t$ since the same proof, with minor modifications, applies to $U_t$.
For $p=\infty$ the statement holds without the requirement on $M^R$. We now consider $p\in [1,\infty)$.
The conditional quantile has the monotonicity property 
\begin{align*}
F_{t,Y}^{-1}(\omega,u)\leq F_{t,\widetilde{Y}}^{-1}(\omega,u)
\quad \text{if } Y\leq \widetilde{Y}
\end{align*} 
and the properties $F_{t,Y}^{-1}(\omega,u)\leq 0$ if $Y\leq 0$ and $F_{t,Y}^{-1}(\omega,u)\geq 0$ if $Y\geq 0$.
In particular, $R_t(|Y|)\leq R_t(Y)\leq R_t(-|Y|)$. We show that $R_t(|Y|),R_t(-|Y|)\in L^p(\calF_t)$.
\begin{align*}
\E{R_t(-|Y|)^p}&=\E{\Big(\int_0^1F_{t,|Y|}^{-1}(u)dM^R(u)\Big)^p}\\
&\leq \E{\int_0^1\big(F_{t,|Y|}^{-1}(u)\big)^pdM^R(u)}\\
&=\E{\int_0^{u_0}F_{t,|Y|^p}^{-1}(u)dM^R(u)}+\E{\int_{1-u_0}^1F_{t,|Y|^p}^{-1}(u)dM^R(u)}\\
&=:E_1+E_2,
\end{align*}
where the first inequality above is an application of Jensen's inequality, and the second equality follows from the fact that $g(F_Z^{-1}(u))=F_{g(Z)}^{-1}(u)$ for increasing functions $g$. Moreover,
\begin{align*}
E_1&\leq \E{F_{t,|Y|^p}^{-1}(u_0)\int_0^{u_0}dM^R(u)}\leq \E{F_{t,|Y|^p}^{-1}(u_0)},\\
E_2&\leq \overline{m}\E{\int_{1-u_0}^1F_{t,|Y|^p}^{-1}(u)du}
\leq \overline{m}\E{\int_0^1F_{t,|Y|^p}^{-1}(u)du}.
\end{align*}
Since 
\begin{align*}
\E{\int_0^1F_{t,|Y|^p}^{-1}(u)du}=\E{\E{|Y|^p\mid\calF_t}}=\E{|Y|^p}<\infty
\end{align*}
and 
\begin{align*}
\E{\int_0^1F_{t,|Y|^p}^{-1}(u)du}\geq \E{\int_{1-u_0}^1F_{t,|Y|^p}^{-1}(u)du}
\geq u_0\E{F_{t,|Y|^p}^{-1}(u_0)},
\end{align*}
$E_1,E_2<\infty$, from which $\E{|R_t(-|Y|)|^p}=\E{R_t(-|Y|)^p}<\infty$ follows.
The argument for showing $\E{|R_t(|Y|)|^p}<\infty$ is completely analogous upon writing
\begin{align*}
\E{(-R_t(|Y|))^p}&=\E{\Big(\int_0^1(-F_{t,-|Y|}^{-1}(u))dM^R(u)\Big)^p}\\
&=\E{\Big(\int_0^1F_{t,|Y|}^{-1}(1-u)dM^R(u)\Big)^p}
\end{align*}
which holds since for every $\omega$, $-F_{t,-|Y|}^{-1}(\omega,u)\neq F_{t,|Y|}^{-1}(\omega,1-u)$ for at most countably many $u\in (0,1)$.

Set $Y:=Y^{(1)}+Y^{(2)}Y^{(3)}$, where $Y^{(1)}\in L^p(\calF_t)$, $0<Y^{(2)}\in L^p(\calF_t)$, and $Y^{(3)}\in L^p(\calF_{t+1})$. Then 
\begin{align*}
F_{t,Y}(\omega,u)&:=Q_{t,Y}(\omega,(-\infty,y])\\
&=Q_{t,Y^{(3)}}(\omega,(-\infty,(y-Y^{(1)}(\omega))/Y^{(2)}(\omega)]).
\end{align*}
Therefore,
\begin{align*}
F_{t,Y}^{-1}(\omega,u)&:=\min\{y\in\R:Q_{t,Y}(\omega,(-\infty,y])\geq u\}\\
&=\min\{y\in\R:Q_{t,Y^{(3)}}(\omega,(-\infty,(y-Y^{(1)}(\omega))/Y^{(2)}(\omega)])\geq u\}\\
&=Y^{(1)}(\omega)+Y^{(2)}(\omega)\min\{y\in\R:Q_{t,Y^{(3)}}(\omega,(-\infty,y])\geq u\}\\
&=Y^{(1)}(\omega)+Y^{(2)}(\omega)F_{t,Y^{(3)}}^{-1}(\omega,u).
\end{align*}
Similarly, $F_{t,-Y}^{-1}(\omega,u)=-Y^{(1)}(\omega)+Y^{(2)}(\omega)F_{t,-Y^{(3)}}^{-1}(\omega,u)$.
It now follows from the definitions of $R_t$ and $U_t$ in \eqref{eq:quantileR} and \eqref{eq:quantileU} that the properties \eqref{eq:ti_r}-\eqref{eq:ph_r} and \eqref{eq:ti_u}-\eqref{eq:ph_u} hold.
The proof of statement (i) is complete.

(ii) Under the stronger assumption that $Y^{(1)}\in L^p(\calF_t)$, $0<Y^{(2)}\in L^0(\calF_0)$, and $Y^{(3)}\in L^p(\calF_{t+1})$ is independent of $\calF_t$,
\begin{align*}
F_{t,Y}^{-1}(\omega,u)&:=Y^{(1)}(\omega)+Y^{(2)}(\omega)F_{t,Y^{(3)}}^{-1}(\omega,u)\\
&=Y^{(1)}(\omega)+Y^{(2)}F_{Y^{(3)}}^{-1}(u)
\end{align*}
since $Y^{(2)}$ is a constant and $Y^{(3)}$ does not depend on $\calF_t$. Similarly, 
\begin{align*}
F_{t,-Y}^{-1}(\omega,u)=-Y^{(1)}(\omega)+Y^{(2)}F_{-Y^{(3)}}^{-1}(u).
\end{align*} 
It follows from the definitions of $R_t$, $U_t$ and $W_t$ in \eqref{eq:quantileR}, \eqref{eq:quantileU} and \eqref{eq:vy}, that $W_t(Y)=Y^{(1)}+Y^{(2)}W_t(Y^{(3)})$, and similarly for $R_t(Y)$ and $U_t(Y)$,
where $W_t(Y^{(3)})$ is a constant if $0<\eta_t\in L^0(\calF_0)$. 
\end{proof}

\begin{remark}
Notice that the condition on $M^R$ in Proposition \ref{prop:quantileRandUandW} holds e.g.~if $M^R$ has support in $(0,1)$ bounded away from $0$ and $1$ (the case for a conditional version of Value-at-Risk) or if $M^R$ has a bounded density (the case for a conditional version of Expected Shortfall).
\end{remark}

\begin{example}\label{ex:spectral}
Here we derive an expression for $W_t(Y_{t+1})$ for $R_t$ and $U_t$ of the form \eqref{eq:quantileR} and \eqref{eq:quantileR}, respectively.
Using well-known properties of quantile functions, we may write
\begin{align*}
&\int_0^1F_{t,(R_t(-Y_{t+1})-Y_{t+1})_+}^{-1}(u)dM^U(u)\\
&\quad=-\int_0^1F_{t,-(R_t(-Y_{t+1})-Y_{t+1})_+}^{-1}(1-u)dM^U(u)\\
&\quad=\int_0^1(R_t(-Y_{t+1})-F_{t,Y_{t+1}}^{-1}(1-u))_+dM^U(u)\\
&\quad =R_t(-Y_{t+1})\int_{1-\gamma_{t}}^1dM^U(u)
-\int_{1-\gamma_{t}}^1F_{t,Y_{t+1}}^{-1}(1-u)dM^U(u),
\end{align*}
where $\gamma_{t}:=\P(Y_{t+1}\leq R_t(-Y_{t+1})\mid\calF_t)$. Hence,
\begin{align*}
W_t(Y_{t+1})&=\Big(1-\frac{1}{1+\eta_t}\int_{1-\gamma_{t}}^1dM^U(u)\Big)\int_0^1F_{t,Y_{t+1}}^{-1}(u)dM^R(u)\\
&\quad+\frac{1}{1+\eta_t}\int_{1-\gamma_{t}}^1F_{t,Y_{t+1}}^{-1}(1-u)dM^U(u).
\end{align*}
If $dM^R(u)=m^R(u)du$ and $dM^U(u)=m^U(u)du$ for monotone integrable functions $m^R$ and $m^U$ with $m^R$ nondecreasing and $m^U$ nonincreasing, then 
\begin{align*}
W_t(Y_{t+1})&=\int_0^1F_{t,Y_{t+1}}^{-1}(u)w_t(u)du,\\
w_t(u)&:=\Big(1-\frac{1}{1+\eta_t}\int_0^{\gamma_{t}}m^U(1-u)du\Big)m^R(u)\\
&\quad+\frac{1}{1+\eta_t}m^U(1-u)1_{[0,\gamma_{t}]}(u).
\end{align*} 
Notice that $w_t$ is nonnegative but not monotone. Similarly to the argument in the proof of Theorem 4.1 in \cite{Acerbi-02}, it follows that $W_t$ does not have the subadditivity property $W_t(Y_{t+1}+\widetilde{Y}_{t+1})\leq W_t(Y_{t+1})+W_t(\widetilde{Y}_{t+1})$. However, see Proposition \ref{prop:gaussian_subadditivity} below, we may ensure subadditivity of the cost-of-capital margin by imposing restrictions on the stochastic model for the liability cash flow.
\end{example}

\section{Cost-of-capital margin for specific models}\label{sec:CoCM_spec_models}

In order to obtain stronger results we need to impose further assumptions. We will therefore assume 
conditional risk measures $R_t$ and conditional utility functions $U_t$ of the kind presented in Section \ref{sec:nice_mappings}. Moreover, we will consider flexible models that, when combined with those conditional risk measures and utility functions provide e.g.~explicit formulas for the cost-of-capital margin. More specifically, in Section \ref{sec:autoregressive} we consider residual cash flows of an autoregressive form, and in \ref{sec:gaussian} we consider a class of Gaussian models for the residual cash flows. 

\subsection{Autoregressive cash flows}\label{sec:autoregressive}

Residual cash flows that are given by an autoregressive process of order one are particularly well suited for explicit computation of the cost-of-capital margin when the dynamic monetary risk measures and dynamic monetary utility functions are of the type presented in Section \ref{sec:nice_mappings}. The autoregressive processes include residual cash flows with independent components as a special case.

\begin{proposition}\label{prop:AR_valuation}
Fix $p\in [1,\infty]$ and let $W_t$ be given by \eqref{eq:vy} with $R_t$ and $U_t$ in \eqref{eq:quantileR} and \eqref{eq:quantileU}, respectively, satisfying the condition in Proposition \ref{prop:quantileRandUandW}.
Let $(Z_t)_{t=1}^T$ be an $\filF$-adapted sequence of random variables such that, for each $t$, $Z_{t+1}\in L^p(\calF_{t+1})$ is independent of $\calF_t$.
Let
\begin{align*}
X_0:=0, \quad X_{t+1}:=\alpha_{t+1}X_{t}+Z_{t+1}, \quad t=0,\dots,T-1,
\end{align*}
and set
\begin{align*}
&\beta_{T}:=1,\quad
\beta_{t}:=1+\beta_{t+1}\alpha_{t+1}, \quad t\in\{1,\dots,T-1\},\\
&\delta_T:=0,\quad
\delta_{t}:=\delta_{t+1}+|\beta_{t+1}|W_{t}(\sign(\beta_{t+1})Z_{t+1}), \quad t\in\{0,\dots,T-1\}.
\end{align*}
Then, for $t=0,\dots,T-1$, $\delta_t\in L^0(\calF_0)$ and 
\begin{align}
V_{t}(X)&=\delta_{t}+\beta_{t+1}\alpha_{t+1}X_{t}\in L^p(\calF_t). \label{eq:AR_valuation}
\end{align}
In particular,
\begin{align*}
V_0(X)=\sum_{t=0}^{T-1}|\beta_{t+1}|W_{t}(\sign(\beta_{t+1})Z_{t+1}).
\end{align*}
\end{proposition}

\begin{proof}[Proof of Proposition \ref{prop:AR_valuation}.]
The statement is proved by induction.
First, 
\begin{align*}
V_{T-1}(X)=W_{T-1}(X_T)=\alpha_{T}X_{T-1}+W_{T-1}(Z_T)=\delta_{T-1}+\beta_{T}\alpha_{T}X_{T-1},
\end{align*}
i.e.~\eqref{eq:AR_valuation} holds for $t=T-1$.
The recursion step: take $t\in\{1,\dots,T-1\}$ and suppose that $V_{t}(X)$ is given by \eqref{eq:AR_valuation}. Then 
\begin{align*}
X_t+V_t(X)=\delta_t+\beta_tX_t=\delta_t+\beta_t\alpha_tX_{t-1}+\beta_tZ_t.
\end{align*}
From Proposition \ref{prop:quantileRandUandW} it follows that $W_{t-1}(\beta_tZ_t)$ is a constant and $X_t+V_t(X)\in L^p(\calF_t)$.
Moreover,
\begin{align*}
V_{t-1}(X)&=W_{t-1}(X_{t}+V_{t}(X))\\
&=\delta_t+\beta_t\alpha_tX_{t-1}+W_{t-1}(\beta_tZ_{t})\\
&=\delta_t+\beta_t\alpha_tX_{t-1}+|\beta_{t}|W_{t-1}(\sign(\beta_{t})Z_{t})\\
&=\delta_{t-1}+\beta_t\alpha_tX_{t-1}.
\end{align*}
In particular, $V_{t-1}(X)=W_{t-1}(X_{t}+V_{t}(X))\in L^p(\calF_{t-1})$.
We conclude that \eqref{eq:AR_valuation} holds for $t=0,\dots,T-1$.
\end{proof}

\begin{remark}
Notice that in the special case where $(X_t)_{t=1}^T$ has independent components, corresponding to $\alpha_t=0$ for all $t$, $V_t(X)=\sum_{s=t}^{T-1}W_s(X_{s+1})$. In particular, $V_t(X)\in L^0(\calF_0)$ if $\eta_t\in L^0(\calF_0)$ for all $t$.
If $\alpha_t=\alpha\in (-1,1)$ for all $t$, then 
\begin{align*}
0<\beta_{t}=\sum_{j=0}^{T-t}\alpha^j=\frac{1-\alpha^{T-t+1}}{1-\alpha}.
\end{align*}
If further $W_t(Z_{t+1})=W_0(Z_1)$ for all $t$, then $V_0(X)=f(\alpha)W_0(Z_{1})$, where
\begin{align*}
f(\alpha):=\sum_{t=1}^T\sum_{j=0}^{T-t}\alpha^j=\sum_{j=0}^{T}(T-j)\alpha^j=
\frac{\alpha^{T+1}-(T+1)\alpha+T}{(1-\alpha)^2}.
\end{align*}
\end{remark}


\subsection{Gaussian cash flows}\label{sec:gaussian}

The convenient properties of conditional distributions of multivariate normal distributions allow for much  stronger results than what have been possible in the setting considered so far. In what follows, we will therefore derive properties of the cost-of-capital margin in a Gaussian setting. Since the cost-of-capital margin is primarily intended for aggregate cash flows, Gaussian model assumptions 
will in many situations provide a reasonable approximation.


\begin{definition}\label{def:gaussian_filtration}
Let $\Gamma$ be a finite set of Gaussian vectors in $\R^T$ that are jointly Gaussian. Let  
\begin{align*}
\calG_0:=\{\emptyset,\Omega\},\quad
\calG_t:=\Big(\vee_{Z\in\Gamma}\sigma(Z_t)\Big)\vee \calG_{t-1} \quad \textrm{for } t=1,\dots,T.
\end{align*}
$\filG:=(\calG_t)_{t=0}^T$ is called a Gaussian filtration, and, 
if $X\in \Gamma$, then $(X,\filG)$ is called a Gaussian model.
\end{definition}

For a Gaussian model $(X,\filG)$, $X$ is interpreted as a cash flow that may be assigned a value and $\filG$ represents the flow of information used in the valuation of $X$.  
Notice that by Proposition \ref{prop:basic_prop} it is sufficient to only consider zero mean Gaussian cash flows $X$.

Consider a Gaussian model $(X,\filG)$ and let 
\begin{align}\label{eq:gaussianY}
Y=a_0+\sum_{Z\in\Gamma}\sum_{s=1}^Ta^Z_sZ_s \text{ for some } a_0\in \R, a^Z_s\in\R. 
\end{align}
Then, the conditional distribution of $Y$ given $\calG_t$ is given by 
\begin{align}\label{eq:gaussian_conditional_distribution}
\P(Y\in A\mid\calG_t)
&=\P(\E{Y\mid\calG_t}+\Var(Y\mid\calG_t)^{1/2}\epsilon_{t+1}\in A\mid\calG_t),
\end{align}  
where 
$\epsilon_{t+1}$ is $\calG_{t+1}$-measurable, standard normal and independent of $\calG_t$. Moreover,
\begin{align}\label{eq:gaussian_ceYgGt}
\E{Y\mid\calG_t}=b_0+\sum_{Z\in\Gamma}\sum_{s=1}^tb^Z_sZ_s
\text{ for some } b_0\in \R, b^Z_s\in\R,
\end{align}
and, a special feature of conditional Gaussian distributions that is essential here, $\Var(Y\mid\calG_s)\in L^0_+(\calG_0)$. 
These properties ensure that if the mapping $W_t:L^p(\calG_{t+1})\to L^p(\calG_t)$ given by \eqref{eq:vy},  with $\filF:=\filG$ and $\eta_t:=\eta_0$, satisfying the assumptions in Proposition \ref{prop:quantileRandUandW}, and if $Y\in L^p(\calG_{t+1})$ is of the form \eqref{eq:gaussianY}, then, by Proposition \ref{prop:quantileRandUandW} and \eqref{eq:gaussian_ceYgGt}, 
for every $p\in [1,\infty)$,
\begin{align*}
W_t(Y)&=\E{Y\mid\calG_t}+\Var(Y\mid\calG_t)^{1/2}W_t(\epsilon_{t+1})\\
&=\E{Y\mid\calG_t}+\Var(Y\mid\calG_t)^{1/2}W_0(\epsilon_{1}).
\end{align*} 

\begin{assumption}\label{assumptionG}
For a zero mean Gaussian model $(X,\filG)$ and $p\in [1,\infty)$, $(W_t)_{t=0}^{T-1}$ is a sequence of mappings $W_t:L^p(\calG_{t+1})\to L^p(\calG_t)$ as in Proposition \ref{prop:quantileRandUandW} with $\filF:=\filG$ and $\eta_t:=\eta_0>0$. Moreover, $\epsilon_1$ is $\calG_1$-measurable and standard normal.
\end{assumption}

\begin{proposition}\label{prop:gaussian_value}
Let $(X,\filG)$ be a zero mean Gaussian model and suppose that Assumption \ref{assumptionG} holds.
Then, for $t\in\{0,\dots,T-1\}$,
\begin{align*}
V_{t,\filG}(X)=\E{\sum_{s={t+1}}^TX_s\mid\calG_t}
+\sum_{s={t+1}}^T\Var\Big(\E{\sum_{u=s}^{T}X_u\mid\calG_s}\mid\calG_{s-1}\Big)^{1/2}W_0(\epsilon_1).
\end{align*}
Moreover, 
\begin{align*}
V_{0,\filG}(X)&=\sum_{s=1}^{T}\Big(\Var\Big(\sum_{u=s}^{T}X_u\mid\calG_{s-1}\Big)
-\Var\Big(\sum_{u=s}^{T}X_u\mid\calG_{s}\Big)\Big)^{1/2}W_0(\epsilon_1).
\end{align*}
\end{proposition}

Notice that, given the assumptions of Proposition \ref{prop:gaussian_value}, we may express the cost-of-capital margin $V_{0,\filG}(X)$ as 
\begin{align*}
\sum_{s=1}^T\Var\Big(X_s+\E{\sum_{u=s+1}^{T}X_u\mid\calG_s}-\E{\sum_{u=s}^{T}X_u\mid\calG_{s-1}}\mid\calG_{s-1}\Big)^{1/2}W_0(\epsilon_1).
\end{align*}
In particular, the cost-of-capital margin $V_{0,\filG}(X)$ is proportional to the sum of the conditional standard deviations of the errors of the repeated predictions of the sum of the remaining cash flows. 

\begin{proof}[Proof of Proposition \ref{prop:gaussian_value}.]
The statement is proved by induction. Let $c:=W_0(\epsilon_1)$ and $V_t(X):=V_{t,\filG}(X)$.
Clearly, $X_T$ is of the form \eqref{eq:gaussianY} so \eqref{eq:gaussian_conditional_distribution} holds. Therefore, from statement (iii) in Proposition \ref{prop:quantileRandUandW},
\begin{align*}
V_{T-1}(X):=W_{T-1}(X_T)=\E{X_T\mid\calG_{T-1}}+\Var(X_T\mid\calG_{T-1})^{1/2}c.
\end{align*}
Now let $t\leq T-1$ and assume that the expression for $V_t(X)$ holds for $t$. Then $X_t+V_t(X)$ is of the form \eqref{eq:gaussianY} so \eqref{eq:gaussian_conditional_distribution} holds. Therefore, from statement (iii) in Proposition \ref{prop:quantileRandUandW},
\begin{align*}
V_{t-1}(X)&:=W_{t-1}(X_t+V_t(X))\\
&=\E{X_t+V_t(X)\mid\calG_{t-1}}+\Var(X_t+V_t(X)\mid\calG_{t-1})^{1/2}c\\
&=\E{\sum_{s=t}^TX_s\mid\calG_{t-1}}
+\sum_{s={t+1}}^T\Var\Big(\E{\sum_{u=s}^{T}X_u\mid\calG_s}\mid\calG_{s-1}\Big)^{1/2}c\\
&\quad\quad+\Var\Big(X_t+\E{\sum_{s={t+1}}^TX_s\mid\calG_t}\mid\calG_{t-1}\Big)^{1/2}c\\
&=\E{\sum_{s={t}}^TX_s\mid\calG_t}
+\sum_{s={t}}^T\Var\Big(\E{\sum_{u=s}^{T}X_u\mid\calG_s}\mid\calG_{s-1}\Big)^{1/2}c.
\end{align*}
Recall the variance decomposition formula:
for $\calF\subset\calG\subset\calH$ and $Y\in L^2(\calH)$,
\begin{align*}
\Var(Y\mid\calF)=\E{\Var(Y\mid\calG)\mid\calF}+\Var(\E{Y\mid\calG}\mid\calF).
\end{align*}
Applying the variance decomposition formula and using the fact that in the Gaussian case the conditional variance is a constant, we find that, for $s<T$, 
\begin{align*}
&\Var\Big(\E{\sum_{u=s}^{T}X_u\mid\calG_s}\mid\calG_{s-1}\Big)\\
&\quad=\Var\Big(\sum_{u=s}^{T}X_u\mid\calG_{s-1}\Big)-\E{\Var\Big(\sum_{u=s}^{T}X_u\mid\calG_s\Big)\mid\calG_{s-1}}\\
&\quad=\Var\Big(\sum_{u=s}^{T}X_u\mid\calG_{s-1}\Big)-\Var\Big(\sum_{u=s}^{T}X_u\mid\calG_s\Big).
\end{align*}
\end{proof}

Computation of $W_0(\epsilon_1)$ is illustrated in the following example.

\begin{example}\label{ex:W0normal}
Let $\epsilon_1$ be standard normal with distribution and density function $\Phi$ and $\phi$, respectively, and let $U_0(\cdot)=\E{\cdot}$. Then
\begin{align*}
W_0(\epsilon_1)
&=R_0(-\epsilon_1)-\frac{1}{1+\eta_0}\E{(R_0(-\epsilon_1)-\epsilon_1)1_{\{\epsilon_1\leq R_0(-\epsilon_1)\}}}\\
&=\Big(R_0(-\epsilon_1)-\frac{1}{1+\eta_0}\Big(R_0(-\epsilon_1)\Phi(R_0(-\epsilon_1))+\phi(R_0(-\epsilon_1))\Big)\Big)\\
&\leq R_0(-\epsilon_1)\frac{\eta_0}{1+\eta_0},
\end{align*}
where the inequality is due to the Mill's ratio inequalities $\phi(x)x/(1+x^2)\leq 1-\Phi(x)\leq \phi(x)/x$ for $x>0$, see e.g.~\cite{Gordon-41}.
Notice that if $R_0=\VaR_p$, then $R_0(-\epsilon_1)=\Phi^{-1}(1-p)$, and for Expected Shortfall, 
$R_0=\ES_p$, corresponding to $m^R(u)=p^{-1}1_{[1-p,1]}(u)$ in Example \ref{ex:spectral}, $R_0(-\epsilon_1)=p^{-1}\phi(\Phi^{-1}(1-p))$.
\end{example}

Bounds on the cost-of-capital margin can be obtained from Proposition \ref{prop:gaussian_value}.
From the proof of Proposition \ref{prop:gaussian_value}, notice that 
\begin{align}\label{eq:gauss_temp}
V_{t,\filG}(X)-\E{\sum_{s={t+1}}^TX_s\mid\calG_t}=W_0(\epsilon_1)\sum_{s=t+1}^Ta_s^{1/2},
\end{align}
where, for $s=1,\dots,T$,
\begin{align}\label{eq:as}
a_s:=\Var\Big(\sum_{u=t+1}^{T}X_u\mid\calG_{s-1}\Big)-\Var\Big(\sum_{u=t+1}^{T}X_u\mid\calG_s\Big).
\end{align}
In particular,
\begin{align*}
\sum_{s=t+1}^T a_s=\Var\Big(\sum_{s=t+1}^TX_s\mid\calG_t\Big).
\end{align*}
An upper bound on the left-hand side in \eqref{eq:gauss_temp} is found by solving a standard convex optimization problem:
\begin{align*}
\text{maximize } & \sum_{s=t+1}^{T}a_s^{1/2}\\
\text{subject to } & \sum_{s=t+1}^{T}a_s=C, \quad 
a_s\geq 0 
\end{align*}
The concave objective function has a unique maximum for $a_s=C/(T-t)$ for all $s$. Minimizing the objective function over the same (convex) set gives a minimum for $a_{s_0}=C$ for some $s_0\in\{t+1,\dots,T\}$ and $a_s=0$ for $s\neq s_0$. 
We have thus proved the following bounds.

\begin{proposition}\label{prop:gaussian_bounds}
Let $(X,\filG)$ be a zero mean Gaussian model and suppose that Assumption \ref{assumptionG} holds.
Then, 
\begin{align*}
&W_0(\epsilon_1)\Var\Big(\sum_{s=t+1}^TX_s\mid\calG_t\Big)^{1/2}\leq V_{t,\filG}(X)-\E{\sum_{s={t+1}}^TX_s\mid\calG_t}\\
&\quad\leq W_0(\epsilon_1)(T-t)^{1/2}
\Var\Big(\sum_{s=t+1}^TX_s\mid\calG_t\Big)^{1/2}.
\end{align*}
In particular, 
\begin{align*}
W_0(\epsilon_1)\Var\Big(\sum_{s=1}^TX_s\Big)^{1/2}\leq V_{0,\filG}(X)
\leq W_0(\epsilon_1)T^{1/2}
\Var\Big(\sum_{s=1}^TX_s\Big)^{1/2}.
\end{align*}
\end{proposition}

The interpretation of the upper bound for the cost-of-capital margin $V_{0,\filG}(X)$ in Proposition \ref{prop:gaussian_bounds} is as follows.
The Gaussian model $(X,\filG)$ maximizing $V_{0,\filG}(X)$ corresponds to a filtration $\filG$ such that   
\begin{align*}
\Var\Big(\sum_{u=1}^{T}X_u\mid\calG_{s-1}\Big)-\Var\Big(\sum_{u=1}^{T}X_u\mid\calG_s\Big)
=\frac{1}{T}\Var\Big(\sum_{u=1}^{T}X_u\Big),
\end{align*}
i.e.~the uncertainty (variance) in the remaining cash flow is distributed evenly over the length of the cash flow.

The interpretation of the lower bound in Proposition \ref{prop:gaussian_bounds} is as follows. Given a zero mean Gaussian cash flow $X$, such that $\Var(X_1)>0$, the Gaussian model $(X,\filG)$ minimizing $V_{0,\filG}(X)$ corresponds to a filtration $\filG$ such that 
\begin{align*}
\Var\Big(\sum_{u=1}^{T}X_u\mid\calG_{s}\Big)=0\quad \text{for } s\geq 1
\end{align*}
i.e.~the cash flow after time $1$ is completely known at time $1$. In particular, after time $1$ there is no need for capital funds and hence there are no capital costs. This interpretation follows from Proposition \ref{prop:early_info} below.

\begin{proposition}\label{prop:early_info}
Let $(X,\filG)$ and $(X,\widetilde{\filG})$ be two Gaussian models such that, for every $t$,
$\widetilde{\calG}_t=\calG_s$ for some $s\geq t$. 
Suppose further that, for both Gaussian models, Assumption \ref{assumptionG} holds.
Then $V_{0,\filG}(X)\geq V_{0,\widetilde{\filG}}(X)$.
\end{proposition}
\begin{proof}
It is sufficient to consider the case $\widetilde{\calG}_{t_0}=\calG_{t_0+1}$ for some $t_0\geq 1$, and $\widetilde{\calG}_{t}=\calG_{t}$ for $t\neq t_0$. Repeating the argument then yields the conclusion.
Set $b_t:=\Var(\sum_{u=1}^TX_u\mid\calG_t)$, $\widetilde{b}_t:=\Var(\sum_{u=1}^TX_u\mid\widetilde{\calG}_t)$, $a_t:=b_{t-1}-b_t$ and $\widetilde{a}_t:=\widetilde{b}_{t-1}-\widetilde{b}_t$. Then 
\begin{align*}
V_{0,\filG}(X)-V_{0,\widetilde{\filG}}(X)&=W_0(\epsilon_1)\sum_{t=1}^T(a_t^{1/2}-\widetilde{a}_t^{1/2})\\
&=W_0(\epsilon_1)\Big(a_{t_0}^{1/2}+a_{t_0+1}^{1/2}-\widetilde{a}_{t_0}^{1/2}-\widetilde{a}_{t_0+1}^{1/2}\Big)\\
&=W_0(\epsilon_1)\Big(a_{t_0}^{1/2}+a_{t_0+1}^{1/2}-0-(a_{t_0}+a_{t_0+1})^{1/2}\Big)\\
&\geq 0
\end{align*}
due to the subadditivity of $\R_+\ni x\mapsto x^{1/2}\in \R_+$.
\end{proof}

In the Gaussian setting, the cost-of-capital margin is subadditive. If the aggregate liability cash flow is decomposed into a sum of sub-liability cash flows, then the sum of the corresponding cost-of-capital margins dominates the cost-of-capital margin for the aggregate liability cash flow. 

\begin{proposition}\label{prop:gaussian_subadditivity}
Let $((X,\widetilde{X}),\filG)$ be a Gaussian model and
suppose that Assumption \ref{assumptionG} holds.
Then, for $t\in\{0,\dots,T-1\}$,
$V_{t,\filG}(X+\widetilde{X})\leq V_{t,\filG}(X)+V_{t,\filG}(\widetilde{X})$.
\end{proposition}
\begin{proof}
From Proposition \ref{prop:gaussian_value}, 
\begin{align*}
W_0(\epsilon_1)^{-1}\Big(V_{t,\filG}(X+\widetilde{X})-V_{t,\filG}(X)-V_{t,\filG}(\widetilde{X})\Big)
=\sum_{s={t+1}}^T\Delta_{s,T},
\end{align*}
where 
\begin{align*}
\Delta_{s,T}&=\Var\Big(\E{\sum_{u=s}^{T}X_u\mid\calG_s}+\E{\sum_{u=s}^{T}\widetilde{X}_u\mid\calG_s}\mid\calG_{s-1}\Big)^{1/2}\\
&\quad-\Var\Big(\E{\sum_{u=s}^{T}X_u\mid\calG_s}\mid\calG_{s-1}\Big)^{1/2}-\Var\Big(\E{\sum_{u=s}^{T}\widetilde{X}_u\mid\calG_s}\mid\calG_{s-1}\Big)^{1/2}
\end{align*}
The conclusion now follows from the general fact
\begin{align*}
\Var(Z+\widetilde{Z}\mid\calG)\leq (\Var(Z\mid\calG)^{1/2}+\Var(\widetilde{Z}\mid\calG)^{1/2})^2.
\end{align*}
\end{proof}



\section{Valuation of a life-insurance portfolio}\label{life_example}


We will now go through a simple, yet realistic, life-insurance example aimed at illustrating aspects of the cost-of-capital margin.

One of the simplest insurance contracts that is non-trivial and for which it is possible to carry out an exact valuation is a portfolio which at time $0$ consists of $n$ identical and independent term life-insurance contracts.
A term life-insurance contract for a today $x$ year old individual that terminates at latest $T$ years from today is constructed as follows: if the insured individual 
\begin{itemize}
\item dies during year $t+1$ (between times $t$ and $t+1$) for $t\in \{0,\dots,T-1\}$, the amount $1$ is paid to the beneficiary at time $t+1$,
\item is alive after $T$ years, the contract pays nothing.
\end{itemize}
Given the above, at time $0$ there are $N_0 = n$ active contracts of $x$ year olds terminating at time $T$. Moreover, deaths of individuals are assumed to be independent events. Further, if we let $D_{t+1}$ denote the number of deaths during year $t+1$, the dynamics of the number of active contracts at each time can be described as a nested binomial process as follows: $(N_t)_{t=0}^{T}$ is a Markov process with 
$N_{t+1}:=N_t-D_{t+1}$ and $D_{t+1} | N_t \sim\text{Bin}(N_t,q_{x+t})$. Here
\begin{align*}
q_{x+t} := \P(T_x \le t+1 | T_x > t) = 1-\frac{S_x(t+1)}{S_x(t)},\quad t=0,\dots,T-1,
\end{align*}
where $T_x$ is the remaining lifetime of a today $x$ year old, where
\begin{align*}
S_x(u) := \exp\left\{-\int_x^{x+u}\mu_s ds\right\},
\end{align*}
where $\mu_s \geq 0$ is the so-called mortality law or force of mortality. Here we let $\mu_s$ be the Makeham mortality law given by
\begin{align*}
\mu_x := \alpha + \beta\exp\{-\gamma x\},~\alpha,\beta,\gamma > 0.
\end{align*}
In the numerical calculations carried out below we will use $\alpha = 0.001, \beta = 0.000012$ and $\gamma = 0.101314$, corresponding to Swedish mortality table M90 for males. The $q_{x+t}$'s are so-called deferred death probabilities, sometimes denoted by $\phantom{}_{t|1}q_x$.
Appendix \ref{life_calcs} describes how the recursive valuation is formulated for a homogeneous population.

If we instead would assign a value to this liability cash flow using EIOPA's standard procedure, we would let the liability value correspond to the so-called Technical Provisions (TP) given by
\begin{align*}
\TP(D) := \BE_{\mu,1}(D) + \RM(D),
\end{align*}
where, neglecting discounting, 
\begin{align*}
\BE_{\mu,i}(D) := \sum_{j=i}^T \E{D_j;\mu}, \quad i=1,\dots,T,
\end{align*}
where $\E{D_j;\mu}$ denotes expected value of $D_j$ using mortality law $\mu$.
Further,
\begin{align*}
\SCR(D) := \BE_{\widetilde{\mu},1}(D) - \BE_{\mu,1}(D),
\end{align*}
$\widetilde{\mu} = 1.15\mu$ corresponds to the stressed mortality defined by EIOPA, see 
\cite[Article 137]{Commission-del-reg-15},
and the risk margin $\RM$ is finally given by
\begin{align*}
\RM(D) :=\CoC \frac{\SCR(D)}{\BE_{\mu,1}(D)}\sum_{i=1}^T \BE_{\mu,i}(D),
\end{align*}
see \cite[Paragraph 1.114, Method 2]{EIOPA-guidline-tp}. $\CoC$ is the cost-of-capital rate, taken to be $0.06$.
Computations of $\BE_{\widetilde{\mu},1}(D)$ and $\BE_{\mu,i}(D)$ are found in Appendix \ref{life_calcs}.
Figure \ref{fig:exactVsEIOPA}
shows $\BE$ as a function of time to contract expiry for a portfolio consisting of $N_0=1000$ 50 year old Swedish males.

\begin{figure}[ht!]
\begin{center}
\includegraphics[width=0.45\textwidth]{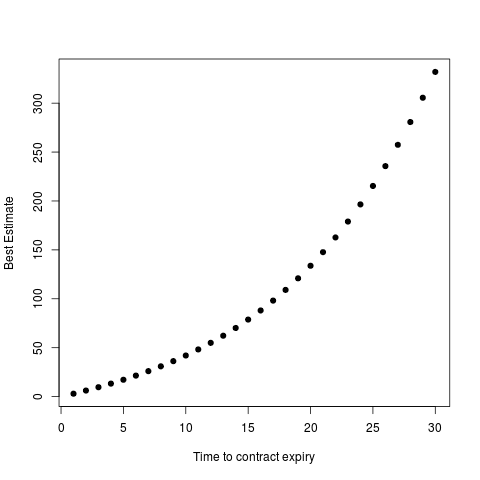}
\includegraphics[width=0.45\textwidth]{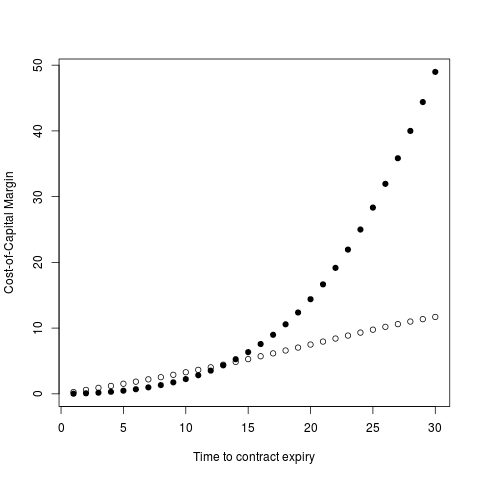}
\end{center}
\caption{The left figure shows best estimate as function of time to contract expiry for a term life-insurance portfolio consisting of $n=1000$ independent $50$ year old Swedish males.
For the same portfolio, the right figure compares the cost-of-capital margins for the nested binomial model (circles) and EIOPA risk margins (solid black discs).}
\label{fig:exactVsEIOPA}
\end{figure}



Figure \ref{fig:exactVsEIOPA} shows a comparison between the valuation according to our interpretation of EIOPA's standard valuation procedure for a portfolio consisting of $N_0=n=1000$ identical 50 year old Swedish males, and the valuation according to 
\eqref{eq:vy} and \eqref{eq:VtWs_rep_def}. Here 
$R_t$ is set to conditional VaR at the 0.5\% level, seen as a mapping $R_t:L^1(\calF_{t+1})\to L^1(\calF_t)$ with $\filF$ taken to be the filtration generated by $(N_t)$, $U_t:=\E{\cdot \mid \mathcal{F}_t}$, and $\eta_t:=\CoC:=0.06$. Information about computational aspects can be found in 
Appendix \ref{sec:nested_binomial}.

For the sake of comparison we focus on EIOPA's risk margin together with the value of the residual cash flow as defined in Section \ref{sec:val_framework}. From Figure \ref{fig:exactVsEIOPA} it can be seen that in this situation EIOPA's risk margin may underestimate as well as overestimate the risk compared to the above more correct valuation procedure. One can argue that the EIOPA method used here is an approximation, but it does not seem to necessarily be a prudent one. This is unfortunate, since the authors believe that this method is commonly used in the industry, given that the so-called proportionality principle applies. 




Further, one can note that the recursions defined in Appendix \ref{life_calcs} are expressed for a single homogeneous population. For heterogeneous populations, numerical problems with computation of binomial probabilities may arise. Therefore, it is of interest to analyse how well the Gaussian approximation of Section \ref{sec:gaussian} performs, since this can be readily adapted to handle heterogeneous populations. In order to do so, we want to compare the cash flow generated by $D$ with that of $\E{D} + X$, where $X$ is a zero mean Gaussian vector with the same covariance matrix as $D$, following the setup of Section \ref{sec:gaussian}. Detailed calculations showing how this is done can be found in Appendix \ref{life_calcs}.

\begin{figure}[ht!]
\begin{center}
\includegraphics[width=0.45\textwidth]{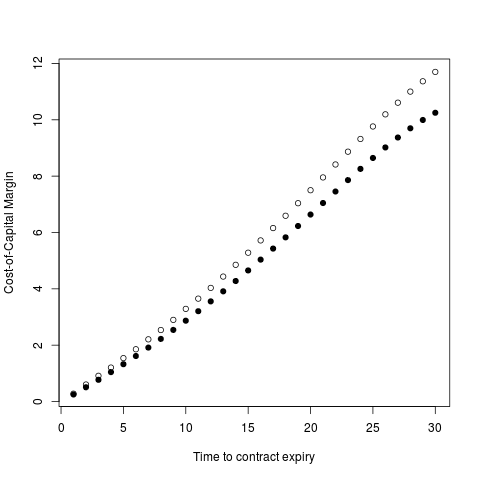}
\includegraphics[width=0.45\textwidth]{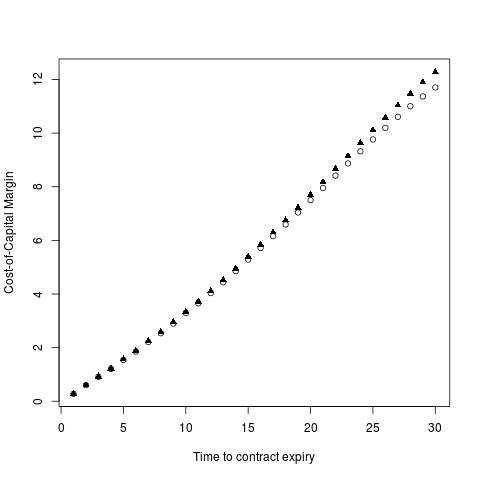}
\end{center}
\caption{For a term life-insurance portfolio consisting of $n=1000$  independent $50$ year old Swedish males, the left figure compares cost-of-capital margins for the nested binomial model (circles) and the Gaussian approximation (solid black discs). The right figure compares cost-of-capital margins for the nested binomial model (circles) and the upper bound \eqref{eq:V0upper_bound_alt} (solid black triangles).}\label{fig:exactVsGauss}
\end{figure}


In Figure \ref{fig:exactVsGauss} we see that the approximation is performing well for the chosen times to expiry of the insurance contracts, given a portfolio size of $N_0=1000$. For larger portfolio sizes and more complex insurance products the Gaussian approximation will serve as a natural benchmark method.


\appendix

\section{Computational details for Section \ref{life_example}}\label{life_calcs}

\subsection{Recursive valuation for a homogeneous population}\label{sec:nested_binomial}

Let $W_t$ be defined as in Section \ref{sec:val_framework}. The backward recursive valuation for the nested binomial model of Section \ref{life_example} can be expressed explicitly as follows.
Write $V_t(D):=G_t(N_t)$,
where $N_t$ denotes the number of active contracts at time $t$, and $G_t$ is some deterministic function, presuming that $\calF_t=\sigma(N_1, ..., N_t)$. With this notation, 
\begin{align*}
Y_{t+1}=D_{t+1}+G_{t+1}(N_{t+1})=D_{t+1}+G_{t+1}(N_{t}-D_{t+1}) 
\end{align*}
and we get the following recursion formula for all $n \in \{0,1,...,N_0\}$: $G_T(n)=0$ and, for $t=0,\dots,T-1$,
\begin{align*}
G_t(n)=W_t(D_{t+1}+G_{t+1}(n-D_{t+1})), \quad D_{t+1} \sim \text{Bin}(n,q_t),
\end{align*}
where the dependence on the age of the insured population has been omitted to simplify the notation.
Starting by calculating $G_{T-1}(n)$ for all feasible values of $n$, we can use the recursive formula until we reach $V_0(D)=G_0(N_0)$. 
The computational feasibility relies on the fact that the relevant information at time $t$ is contained in $N_t$ which takes values in the relatively small set $\{0,\dots,N_0\}$. If we consider information which may be expressed as a vector in $\mathbb{N}^k$, $k\geq 2$, direct computation will be considerably more involved, and possibly unfeasible. This would be the case if the population of insured consisted of $k\geq 2$ homogeneous subgroups.

\subsection{Computation of moments in the nested binomial model}

We will now go through how $\E{D_i}$ and $\E{D_iD_j}$ are calculated for the nested binomial model described in Section \ref{life_example}. In order to ease notation we omit explicit references to the age $x$ of the insured population. 
First, note that with $q_k:=\P(T_x\leq k+1\mid T_x>k)$,
\begin{align*}
&N_i \sim \text{Bin}(N_0,\widetilde{p}_i), \quad \widetilde{p}_i:=\prod\limits_{k=0}^{i-1}(1-q_k), \quad i\geq 1,\\
&D_{i+1}\mid N_l \sim \text{Bin}(N_l, q_{i|l}), \quad q_{i|l}:=q_i \prod\limits_{k=l}^{i-1}(1-q_k), \quad i\geq l, l\geq 0. 
\end{align*}
In particular, $\E{D_{i+1}}=q_{i|0}N_0$. Moreover, for $j>i\geq 0$,
\begin{align*}
\E{D_{i+1}D_{j+1}}&=\E{\E{D_{i+1}D_{j+1} \mid N_i}}\\
&= \sum\limits_{n=0}^{N_0}\P(N_i=n)\sum\limits_{x=0}^{n}\sum\limits_{y=0}^{n-x} xy\P(D_{i+1}=x, D_{j+1}=y\mid N_i = n) \\
&=\sum\limits_{n=0}^{N_0}\binom{N_0}{n}\widetilde{p}_i^n(1-\widetilde{p}_i)^{N_0-n}\\
&\quad\times\sum\limits_{x=0}^{n}\sum\limits_{y=0}^{n-x} xy\P(D_{i+1}=x, D_{j+1}=y\mid N_i = n),
\end{align*}
where
\begin{align*}
&\P(D_{i+1}=x, D_{j+1}=y\mid N_i = n)\\
&\quad=\P(D_{i+1}=x \mid N_i = n)\P(D_{j+1}=y \mid D_{i+1}=x, N_i = n)\\
&\quad=\P(D_{i+1}=x \mid N_i = n)\P(D_{j+1}=y \mid N_{i+1 }= n-x) \\
&\quad=\binom{n}{x}q_i^x(1-q_i)^{n-x}\binom{n-x}{y}q_{j|(i+1)}^y(1-q_{j|(i+1)})^{n-x-y}.
\end{align*}
Combining these expressions, we get
\begin{align*}
\E{D_{i+1}D_{j+1}}&=\sum\limits_{n=0}^{N_0} \binom{N_0}{n}\widetilde{p}_i^n(1-\widetilde{p}_i)^{N-n}\\
&\quad\times\sum\limits_{x=0}^{n}\sum\limits_{y=0}^{n-x}xy \binom{n}{x}q_i^x(1-q_i)^{n-x}\binom{n-x}{y}q_{j|(i+1)}^y(1-q_{j|(i+1)})^{n-x-y}
\end{align*}
and similarly
\begin{align*}
\E{D_{i+1}^2}=\sum\limits_{n=0}^{N_0}\binom{N_0}{n}\widetilde{p}_i^n(1-\widetilde{p}_i)^{N_0-n}\sum\limits_{x=0}^{n}x^2\binom{n}{x}q_i^x(1-q_i)^{n-x}.
\end{align*}

\subsection{Cost-of-capital recursion for Gaussian models}

Here we derive an explicit recursion formula for the cost-of-capital margin for the Gaussian model $(X, \filG)$ with $\filG=(\calG_t)_{t=0}^T$, $\calG_0=\{\emptyset,\Omega\}$ and $\calG_t=\sigma(X_1,...,X_t)$ for $t=1,\dots,T$. 
A similar, albeit more complicated, formula could be derived for a general Gaussian model, meaning that $\filG$ is larger than the natural filtration of $X$. 

For a multivariate normal vector $Z\sim N_{n}({\mu},\Sigma)$, write 
\begin{align*}
\mu=\left[\begin{array}{c}
\mu_{1:n-1}\\
\mu_n
\end{array}
\right],
\quad
\Sigma=\left[\begin{array}{lr}
\Sigma_{1:n-1,1:n-1} & \Sigma_{1:n-1,n}\\
\Sigma_{n,1:n-1} & \Sigma_{n,n}
\end{array}
\right]
\end{align*}
It is well known that the conditional distribution of $Z_n$ given $Z_1,\dots,Z_{n-1}$ is normal with parameters 
\begin{align*}
\mu_{n\mid 1:n-1}&=\mu_n+\Sigma_{n,1:n-1}\Sigma_{1:n-1,1:n-1}^{-1}(Z_{1:n-1}-\mu_{1:n-1}),\\
\Sigma_{n\mid 1:n-1}&=\Sigma_{n,n}-\Sigma_{n,1:n-1}\Sigma_{1:n-1,1:n-1}^{-1}\Sigma_{1:n-1,n}.
\end{align*}


\begin{proposition}\label{prop:gaussianrecursion}
Let $X=(X_t)_{t=1}^{T}\sim N_{T}({0},\Sigma)$ where $\Sigma$ is invertible, let $\filG$ be its natural filtration, and suppose that Assumption \ref{assumptionG} holds.
Then 
\begin{align*}
&V_{t}(X)=\left\{\begin{array}{ll}
0, & t=T,\\
(v^{(t)})^{\trans}X_{1:t} + k_t, & t\in\{1,\dots,T-1\},\\
k_0, & t=0,
\end{array}
\right.
\end{align*}
where $k_{T}:=0$, $v^{(T)}:=0$, $v^{(0)}:=0$,
and, for $t\in\{1,\dots,T-1\}$, $k_t\in\R$ and $v_{t}\in\R^t$ can be calculated recursively from  
\begin{align*}
k_t&:=k_{t+1}+W_t((1+v^{(t+1)}_{t+1})\Sigma_{t+1\mid 1:t}^{1/2}\epsilon_{t+1}),\\
(v^{(t)})^{\trans}&:= (v^{(t+1)}_{1:t})^{\trans}+(1+v^{(t+1)}_{t+1})\Sigma_{t+1:1:t}\Sigma_{1:t,1:t}^{-1}, 
\end{align*}
where $(\epsilon_t)_{t=1}^{T}$ is a sequence of independent standard normally distributed random variables such that $\epsilon_{t+1}$ is $\calG_{t+1}$-measurable and independent of $\calG_t$.
\end{proposition}

\begin{proof}[Proof of Proposition \ref{prop:gaussianrecursion}]
We know that $V_{T}(X)=0$ and we set $v^{(T)}:=0$. 
We prove the statement via induction. Take $t\in\{0,\dots,T-1\}$ and suppose that $V_{t+1}(X)=(v^{(t+1)})^{\trans}X_{1:t+1}+k_{t+1}$. 
First, consider the case $t\geq 1$. Then 
\begin{align*}
Y_{t+1}=(1+v^{(t+1)}_{t+1})X_{t+1}+(v^{(t+1)}_{1:t})^{\trans}X_{1:t}+k_{t+1}.
\end{align*}
Since the latter two terms are $\calG_t$-measurable, translation invariance of $W_t$ combined with properties of the conditional Gaussian distribution give
\begin{align*}
V_t(X)&=W_t(Y_{t+1})\\
&=W_t((1+v^{(t+1)}_{t+1})X_{t+1})+(v^{(t+1)}_{1:t})^{\trans}X_{1:t}+k_{t+1}\\
&=W_t((1+v^{(t+1)}_{t+1})(\mu_{t+1\mid 1:t}+\Sigma_{t+1\mid 1:t}^{1/2}\epsilon_{t+1}))+(v^{(t+1)}_{1:t})^{\trans}X_{1:t}+k_{t+1}\\
&=W_t((1+v^{(t+1)}_{t+1})\Sigma_{t+1\mid 1:t}^{1/2}\epsilon_{t+1})+(1+v^{(t+1)}_{t+1})\mu_{t+1\mid 1:t}+(v^{(t+1)}_{1:t})^{\trans}X_{1:t}+k_{t+1}\\
&=\Big((v^{(t+1)}_{1:t})^{\trans}+(1+v^{(t+1)}_{t+1})\Sigma_{t+1,1:t}\Sigma_{1:t,1:t}^{-1}\Big)^{\trans}X_{1:t}\\
&\quad+k_{t+1}+W_t((1+v^{(t+1)}_{t+1})\Sigma_{t+1\mid 1:t}^{1/2}\epsilon_{t+1})\\
&=(v^{(t)})^{\trans}X_{1:t}+k_t
\end{align*}
from which the conclusion follows. Finally, consider the case $t=0$. Then 
\begin{align*}
V_0(X)&=W_0(Y_{1})\\
&=W_0((1+v^{(1)}_{1})X_{1})+k_{1}\\
&=k_1+(1+v^{(1)}_{1})\mu_1+W_0((1+v^{(1)}_{1})\Sigma_{1,1}^{1/2}\epsilon_{1})\\
&=k_0.
\end{align*}
\end{proof}

Notice that $W_t(K \epsilon_{t+1})=|K|W_0(\epsilon_{1})$ due to symmetry of the standard normal distribution and Assumption \ref{assumptionG}. 

\end{document}